\documentclass{article}
\usepackage{iclr2027_conference,times}

\usepackage{microtype}
\usepackage{graphicx}
\usepackage{subcaption}
\usepackage{booktabs}
\usepackage{tabularx}
\usepackage{wrapfig}
\usepackage{amsmath,amssymb,mathtools}
\usepackage{amsthm}
\usepackage{xcolor}
\usepackage{hyperref}
\usepackage{url}
\usepackage[capitalize,noabbrev]{cleveref}

\newtheorem{proposition}{Proposition}
\newtheorem{corollary}{Corollary}
\newtheorem{assumption}{Assumption}
\newtheorem{remark}{Remark}

\newcommand{\adv}{\textit{adv}}

\title{Near-Duplicate Families Break Exact-Record Membership Inference}
\author{Yiyong Liu\textsuperscript{1}, Jiayang Liu\textsuperscript{2}, Yixin Tan\textsuperscript{3}, Lu Sun\textsuperscript{4}, Rui Wen\textsuperscript{3} \\
\textsuperscript{1}CISPA Helmholtz Center for Information Security \\
\textsuperscript{2}Nanyang Technological University \\
\textsuperscript{3}Institute of Science Tokyo \\
\textsuperscript{4}Tohoku University
}
\iclrfinalcopy 

\begin{document}
\maketitle

\begin{abstract}
Membership inference (MI) asks whether a specific record appeared in a model's training set and is increasingly used as evidence for data provenance and copyright auditing. These applications require determining whether the exact queried record was used for training, rather than merely whether the model was exposed to similar content. Making this distinction is challenging because web-scale datasets naturally contain near-duplicates, including syndicated articles, mirrored pages, and lightly modified images.
We show that this creates a fundamental confound for standard MI. A clean-reference audit typically calibrates membership against a null in which neither the queried record nor its near-duplicate family is present. In deployment, however, the queried record may be absent while a non-identical family member was used for training. We introduce a four-world audit that independently varies exact-record inclusion and family presence to separate these cases.
Natural near-duplicate families cause severe false attribution. On CC-News, a clean-reference LiRA auditor labels 99.70\% of family-present exact non-members as members at 1.00\% false-positive rate. This failure persists across alternative scores, model architectures, and executed deduplication and retraining. Controlled interventions further reveal that the effect depends on the learning objective. In classification, faithful families largely substitute for the exact record, reducing exact-given-family inference to near chance. In autoregressive language modeling, the exact sequence retains a detectable residual, while family presence still confounds clean-reference decisions.
These results show that positive model-only membership evidence may establish family-level exposure without establishing exact-record provenance.
\end{abstract}

\section{Introduction}
\label{sec:intro}

Membership inference (MI) asks whether a specific record appeared in a model's training set~\citep{SSSS17,YGFJ18,CCNSTT22,LZBZ22,WLBZ24}.
While initially designed as a privacy attack, MI is increasingly repurposed as forensic evidence in copyright disputes and data provenance claims~\citep{MYP21,MSFM24,EBCMD26}.
These applications demand a rigorous standard of proof: the audit must identify the \emph{exact} disputed record rather than merely related content.
Yet, web-scale corpora inherently blur this distinction through syndicated news, mirrored pages, and minor edits~\citep{LINZECC22,KWR22}.
Standard MI ignores this reality by relying on a ``clean'' null hypothesis that tests exact inclusion against an idealized absence condition where no related content exists.
Consequently, if a model assigns a member-like score due to the presence of non-identical near-duplicates, standard auditors risk triggering a false positive, confounding exact-record provenance with family-level exposure.

We demonstrate that this vulnerability severely compromises exact-record attribution in practice.
On CC-News, natural near-duplicate families cause a clean-reference auditor to falsely label $99.70\%$ of family-present genuine non-members as members after full fine-tuning of a Pythia-2.8B model~\citep{BSABOHKPPRSSW23}.
The effect persists across alternative membership scores, including Min-K\%~\citep{SAXHLBCZ24}, Min-K++~\citep{ZSYOKZYL25}, and EZ-MIA~\citep{ISC26}. It also transfers to different model and fine-tuning settings and persists after an executed MinHash deduplication-and-retraining pipeline.
Similarly, natural WikiText-103 siblings independently increase the false-member rate from $1.40\%$ to $81.00\%$.
In both corpora, no-sibling control records remain calibrated at roughly their nominal rates ($1$--$2\%$). This isolates family presence, rather than general threshold miscalibration, as the source of the failure.
A standard positive MI score can therefore provide evidence of broader content exposure without proving exact-record inclusion.

To formally isolate what membership signal, if any, survives this confounding, we introduce a four-world audit (\cref{fig:fourworld}).
This framework independently varies exact-record inclusion ($E$) and near-duplicate family presence ($F$), separating standard exact-record inference from exact-given-family inference.
Applying this audit in controlled settings reveals a fundamental, objective-dependent split in the underlying mechanics.
In classification tasks, faithful near-duplicate families act as substitutes for the exact example.
This redundancy drives exact-record detectability down to random chance~\citep{F20}.
Conversely, autoregressive language modeling retains an additive, sequence-specific effect: a language model preserves a unique, detectable signature of the exact text even when its family is present~\citep{CTWJHLRBSEOR21,FZ20}.
This surviving signature nevertheless cannot save a standard audit: family presence alone inflates the score past any clean-calibrated threshold, triggering a false positive before the exact-sequence residual is ever weighed.
Model-only audits should therefore declare their unit of inference and abstain on exact-record verdicts whenever near duplicates cannot be excluded, while still reporting family-level exposure, which remains reliably detectable (AUC $0.98$ on CC-News).

\begin{figure*}[t]
  \centering
  \begin{subfigure}[t]{0.315\linewidth}
    \centering
    \includegraphics[width=\linewidth]{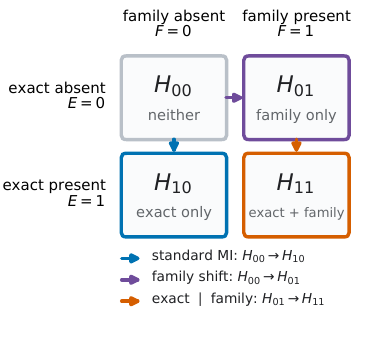}
    \caption{Audit design.}
  \end{subfigure}\hfill
  \begin{subfigure}[t]{0.315\linewidth}
    \centering
    \includegraphics[width=\linewidth]{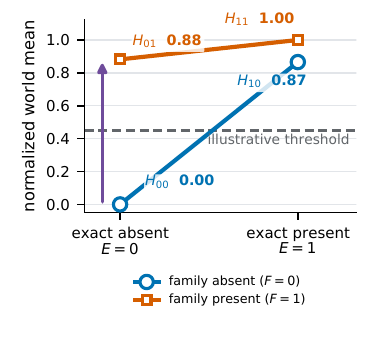}
    \caption{Classification: substitution.}
  \end{subfigure}\hfill
  \begin{subfigure}[t]{0.315\linewidth}
    \centering
    \includegraphics[width=\linewidth]{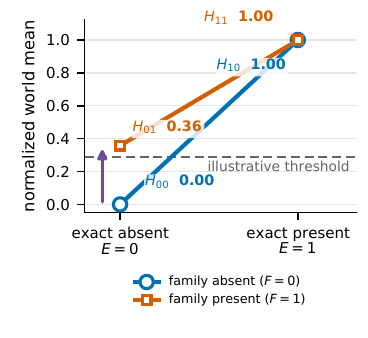}
    \caption{LM: residual exact signal.}
  \end{subfigure}
  \caption{The four-world audit isolates exact-record evidence from family presence. \textbf{(a)} Standard MI (blue, $H_{00}\!\to\!H_{10}$) tests exact inclusion against a clean absence world, ignoring the family-only shift (purple, $H_{00}\!\to\!H_{01}$); our audit instead toggles the exact record with the family already present (orange, $H_{01}\!\to\!H_{11}$). \textbf{(b)} Under classification the family substitutes, flattening the orange trajectory. \textbf{(c)} Under autoregressive modeling the exact record keeps a large additive residual, while the family shift alone (purple) already crosses the illustrative threshold (dashed). Panels (b,c): min--max-normalized pooled world means (matched AG-News/GPT-2 shadows); lines connect binary $E$ states.}
  \label{fig:fourworld}
\end{figure*}

This paper makes three primary contributions:
\begin{itemize}
  \item We formulate exact-record membership as a composite hypothesis with near-duplicate family presence as a critical nuisance variable, and introduce a four-world intervention that cleanly separates exact-record evidence from family-level evidence.
  \item We empirically demonstrate severe natural false attribution across multiple large-scale text corpora, showing this vulnerability survives executed deduplication pipelines and transfers across different models and attack formulations.
  \item We identify a fundamental split in learning dynamics between substitution under classification and exposure-additive memory under language modeling, establishing that while exact provenance is empirically unidentifiable in these settings, family presence remains a robust and defensible conclusion.
\end{itemize}

\section{Related Work}
\label{sec:related}

\paragraph{Membership inference and attribution.}
Membership inference originated as a privacy attack closely tied to model overfitting~\citep{SSSS17,YGFJ18}.
Later work showed how little access the attack needs and how far it reaches: membership leaks through hard labels alone~\citep{LZ21}, sharpens when the attacker recovers training dynamics through distilled loss trajectories~\citep{LZBZ22}, tracks architectural choices such as multi-exit networks~\citep{LLHYBZ22}, and persists even when adaptation happens in context rather than in model weights~\citep{WLBZ24}.
Large-scale, low-FPR formulations such as LiRA~\citep{CCNSTT22} and RMIA~\citep{ZLS24}, alongside LLM-specific metrics like Min-K\%~\citep{SAXHLBCZ24,ZSYOKZYL25,ISC26}, have since sharpened its statistical power.
As these tools are increasingly repurposed for dataset inference and copyright analysis~\citep{MJPD24,MYP21,MSFM24,EBCMD26}, the validity of their underlying assumptions has come under scrutiny.
While recent work highlights the calibration obstacle of unsamplable production nulls~\citep{ZDKT25}, we identify a more fundamental flaw that persists even with perfect retraining access: the ``absence'' of a queried record is a composite state containing both family-absent and family-present conditions, which severely confounds exact attribution.

\paragraph{Related records and memorization.}
Prior work establishes that memorization is fundamentally linked to data distribution~\citep{FZ20,F20,CTWJHLRBSEOR21}, with neighboring non-members often receiving member-like predictions~\citep{RL22}.
Furthermore, ``privacy onion'' and counterfactual memorization frameworks demonstrate that a record's leakage is deeply dependent on its surrounding training neighborhood~\citep{CJZPTT22,ZILJTC23}, motivating standard industry practices like corpus deduplication to mitigate extraction risks~\citep{KWR22,LINZECC22}.

\paragraph{Manipulation and defenses.}
Optimized poisons can force arbitrary MI decisions~\citep{MHLMFJPC25,BJLMSC25}, and active attacks like Truth Serum and Chameleon amplify leakage through mislabeled copies or adaptive poisoning~\citep{TSJLJHC22,CSOU24}.
Trainer-side defenses instead mitigate leakage through differential privacy or adversarial regularization~\citep{ACGMMTZ16,NSH18,TMSSNHM22,CYF22,JSBZG19}.
Unlike these active interventions, the confounding effect we evaluate arises natively without an adversary, and any defensive interpretation of our controlled experiments remains strictly secondary.

\section{Near-Duplicate Families Confound Exact-Record Inference}
\label{sec:bg}
\label{sec:method}

\subsection{The four-world audit framework}

Standard MI treats membership as a binary question: whether the exact record was used for training.
With near-duplicates, however, the observed membership signal may come from either the exact record or related family members.
We therefore separate these two factors and use a four-world audit framework to distinguish three questions: standard exact-record inference, family-presence inference, and exact-record inference when the family is already present.

Let $f$ be a model trained on dataset $D$, and let $x$ be the exact record whose membership is being audited.
Let $E\in\{0,1\}$ indicate whether $x$ is present in $D$, and let $F\in\{0,1\}$ indicate whether a non-identical member of its near-duplicate family is present.
The pair $(E,F)$ defines four worlds, $H_{00}$, $H_{10}$, $H_{01}$, and $H_{11}$.

Standard MI compares $H_{10}$ with $H_{00}$.
Exact-given-family inference compares $H_{11}$ with $H_{01}$ and asks whether $x$ adds detectable evidence after its family is already present.
Family-presence inference compares $H_{01}$ with $H_{00}$ and asks whether the family itself leaves detectable evidence.
When $F$ is unknown, exact-record inference has the composite hypotheses
\begin{equation}
  \mathcal H_0^E=\{H_{00},H_{01}\},
  \qquad
  \mathcal H_1^E=\{H_{10},H_{11}\}.
  \label{eq:composite}
\end{equation}
A clean-reference test calibrated on $H_{00}$ controls false positives for only one component of the null.

\begin{proposition}[Clean-null calibration]
Let $A$ be any rejection region with $P_{00}(A)\leq\alpha$.
Its type-I error for the composite exact-absence null is
\begin{equation}
  \sup_{F\in\{0,1\}}P_{0F}(A)=\max\{P_{00}(A),P_{01}(A)\},
  \label{eq:worstcase-fpr}
\end{equation}
which is not bounded by $\alpha$ and may equal one.
If family presence among exact non-members has prevalence $\pi$, the deployment FPR is $(1-\pi)P_{00}(A)+\pi P_{01}(A)$.
\end{proposition}
\begin{proof}
Both statements follow directly by maximizing, or averaging, over the two components of $\mathcal H_0^E$.
Calibration on $H_{00}$ constrains only the first term, leaving the error rate on $H_{01}$ unbounded.
\end{proof}

The proposition is distribution-free: it does not depend on LiRA, Gaussian scores, or a particular similarity metric.
The empirical question is whether $P_{01}(A)$ differs materially from $P_{00}(A)$; our natural studies show that it can.

The four comparisons separate questions that a binary MI benchmark conflates.
Standard MI measures whether $x$ has an effect when no related content is present, but it does not identify the source of a member-like score in a corpus where $F$ is unknown.
Exact-given-family inference is the relevant controlled comparison for exact provenance: both hypotheses contain the same family and differ only in $x$.
Conversely, family-presence inference deliberately removes $x$ and asks whether its surrounding content is visible.
This distinction also prevents a low exact-given-family AUC from being misread as content privacy; the family itself may remain highly detectable.

For a score with four-world means $\mu_{EF}$ and variances $\sigma_{EF}^2$, we define the exact record's raw conditional effect and its standardized effect as:
\begin{equation}
  \delta_E^F=\mu_{1F}-\mu_{0F}, \qquad
  d_E^F=\frac{\delta_E^F}{\sqrt{(\sigma_{1F}^2+\sigma_{0F}^2)/2}}.
  \label{eq:factorial}
\end{equation}
The factorial interaction $I = \delta_E^1 - \delta_E^0$ measures how family presence shifts the exact record's absolute contribution.
The ratio of the standardized gaps, $R_{\mathrm{exact}} = d_E^1 / d_E^0$, measures the retained fraction of the exact-record signal.
A ratio near zero indicates near-total substitution, while a ratio near one indicates approximate additivity.

\subsection{Core failure modes and baseline auditors}

We distinguish two distinct failure modes under near-duplicate exposure:
\begin{itemize}
  \item \textbf{False attribution} occurs when the presence of the family alone ($H_{01}$) causes model scores to cross a threshold calibrated on clean data ($H_{00}$ vs. $H_{10}$), triggering a false positive for the absent exact record.
  \item \textbf{Erasure} occurs when the exact record's signal is completely overshadowed by the family, rendering the worlds $H_{11}$ and $H_{01}$ empirically indistinguishable even to an oracle auditor.
\end{itemize}

The auditor has black-box access to the target model and may train shadow models on reference datasets, but cannot inspect the deployed training set or assume knowledge of every variant in the near-duplicate family.
For classification, the primary score is the true-class logit margin.
For language modeling, it is negative suffix loss.
Our primary calibrated auditor is LiRA~\citep{CCNSTT22}, which fits parametric score densities on calibration shadows to evaluate the log-likelihood ratio:
\begin{equation}
  \ell_x(\phi) = \log p(\phi \mid E=1, F=f) - \log p(\phi \mid E=0, F=f).
  \label{eq:lira}
\end{equation}
The conditioning state $F=f$ defines our two primary auditor baselines: the \emph{clean-reference auditor} (calibrated assuming the family is absent, $F=0$) and the \emph{family-conditioned auditor} (calibrated knowing the family is present, $F=1$).
Standard clean-reference decision thresholds are chosen at a nominal false-positive rate $\alpha$ using clean OUT scores ($H_{00}$).
False attribution then occurs precisely because the actual deployment state shifts to $F=1$ (the family is present), blinding the clean-reference auditor to the difference.

\subsection{Evaluation protocol}

Our evaluation protocol follows a rigorous shadow-training pipeline.
Each condition trains a fresh set of shadow models while independently randomizing each target's exact membership.
Calibration shadows fit the conditional score densities and establish the decision thresholds; disjoint evaluation shadows are then used to measure false-positive and true-positive rates out of sample.
We report the raw Area Under the ROC Curve (AUC), the orientation-invariant advantage $\adv=2|\mathrm{AUC}-0.5|$, and the decision rule's positive likelihood ratio $\mathrm{LR}^{+}=\mathrm{TPR}/\mathrm{FPR}$ (computed as the bootstrap-median ratio over individual resamples to handle skewness).
Confidence intervals are estimated using a crossed bootstrap over both evaluation targets and shadow seeds.
To prevent selection bias, vulnerable targets are identified solely on a separate set of discovery shadows.

We evaluate this protocol across both text and vision modalities:
\begin{itemize}
  \item \textbf{Text corpora:} Near-duplicate families are mined naturally from CC-News and WikiText-103 using MinHash clustering over token $5$-grams with a Jaccard similarity threshold~\citep{LINZECC22,KWR22}.
  \item \textbf{Vision corpora:} Families are systematically constructed on CIFAR-10/100 using same-label spatial and geometric transformations of the target images.
\end{itemize}
Appendix~\ref{app:expmap} provides a complete map of every reported quantity to its exact construction, score, subset, and estimator.

\paragraph{Experimental scope.}
Our evaluation encompasses multiple modalities (vision, text), architectures (ResNet-9~\citep{HZRS16}, VGG-11~\citep{SZ15}, GPT-2, Pythia-2.8B~\citep{BSABOHKPPRSSW23}, Qwen3.5-2B~\citep{T25}), and datasets (CIFAR-10/100, AG-News, WikiText-103~\citep{MXBS17}, CC-News), requiring the training of over $1{,}500$ shadow models in total.
Every intervention holds the background data and training recipe fixed, isolating the exact-record signal by independently randomizing membership.
Comprehensive data construction, optimization, and compute configurations for each setting are deferred to the appendix.

\section{Natural Families Cause Severe False Attribution}
\label{sec:natural}

\subsection{Widespread false attribution on CC-News}
\label{sec:scaleup}

We evaluate false attribution in the wild by mining $200$ natural news-syndication families and $100$ no-sibling controls from CC-News.
The targets span broad token-$5$-gram Jaccard similarities ($[0.50, 1.00)$), ensuring the audit tests non-identical, natural variation.
Across $128$ pythia-2.8B shadows, we establish a highly vulnerable baseline: the family-absent exact-record advantage is $0.99$ (averaged per target; the pooled score AUC in \cref{tab:fourworld} is $0.87$), meaning standard MI succeeds almost perfectly in a clean world.
To rigorously isolate the family effect, we score a fixed suffix of each passage, remove byte-identical copies, and evaluate against the no-sibling controls processed through the identical calibration pipeline.
Because exact membership is independently randomized in every shadow, any pretraining contamination remains strictly symmetric between the IN and OUT worlds, isolating the observed failure entirely to the family-present contrast.

\paragraph{Severe false attribution.}
The clean-reference per-target LiRA auditor nevertheless labels $99.70\%$ of family-present non-members as members at a nominal $1.00\%$ FPR (\cref{fig:natural-verdict}); the rate remains $99.70\%$ at the $5.00\%$ operating point used in \cref{tab:natural-robustness}.
Its positive likelihood ratio falls from approximately $62.00$ in the clean world to $1.00$, because TPR and FPR become nearly equal.
In the matched pooled control, no-sibling targets have only a $1.20\%$ false-positive rate at that operating point.
The effect is thus a family-specific loss of exact attribution rather than indiscriminate miscalibration.

\begin{figure*}[t]
  \centering
  \begin{subfigure}[t]{0.315\linewidth}
    \centering
    \includegraphics[width=\linewidth]{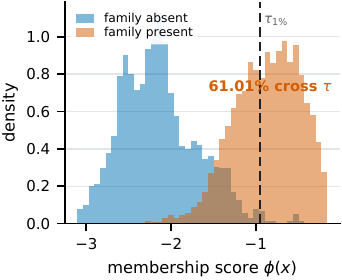}
    \caption{Non-member score shift.}
  \end{subfigure}\hfill
  \begin{subfigure}[t]{0.315\linewidth}
    \centering
    \includegraphics[width=\linewidth]{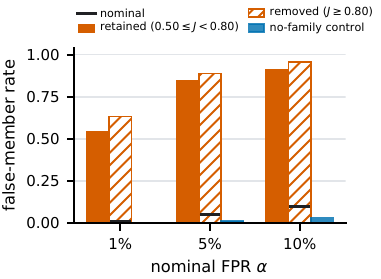}
    \caption{Jaccard-$0.8$ retention bands.}
  \end{subfigure}\hfill
  \begin{subfigure}[t]{0.315\linewidth}
    \centering
    \includegraphics[width=\linewidth]{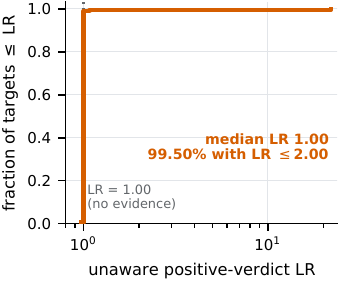}
    \caption{Per-target evidence.}
  \end{subfigure}
  \caption{Natural CC-News/pythia-2.8B verdicts. (a) Family presence shifts pooled non-member scores: $61.01\%$ cross a global clean $1.00\%$-FPR threshold (per-target calibrated FPR $99.70\%$). (b) False attribution stays high even in the Jaccard-$0.80$ dedup-retained band. (c) The unaware positive-verdict LR has median $1.00$; $199/200$ targets have LR $\leq2.00$.}
  \label{fig:natural-verdict}
\end{figure*}

\paragraph{Executed deduplication and retraining.}
We execute a Jaccard-$0.80$ MinHash deduplication pass, rebuild the corpus, and retrain rather than merely partitioning a fixed run by similarity.
Sub-threshold families retained by that pipeline still produce false-member rates of approximately $92.00\%$ and $99.00\%$ in the $[0.50,0.70)$ and $[0.70,0.80)$ bands at nominal $5.00\%$ FPR.
Deduplication reduces repeated content but cannot certify that every provenance-confounding family is absent.

\begin{wraptable}{R}{0.52\linewidth}
  \centering
  \footnotesize
  \renewcommand{\arraystretch}{1.06}
  \setlength{\tabcolsep}{1.3pt}
  \vspace{-4mm}
  \caption{CC-News robustness at nominal $5.00\%$ FPR. FP is the false-attribution rate on family-present exact non-members (no-sibling records for the control row). P/T denote pooled/per-target calibration.}
  \label{tab:natural-robustness}
  \begin{tabularx}{\linewidth}{@{}l>{\raggedright\arraybackslash}Xlr@{}}
    \toprule
    Model/run & Score (cal.) & Targets & FP \\
    \midrule
    Pythia/core & loss (P) & all & $88.00\%$ \\
                 & LiRA (T) & all & $99.70\%$ \\
                 & loss (P) & control & $1.20\%$ \\
    \addlinespace[1pt]
    Pythia/repeat & loss (P) & all & $98.00\%$ \\
                   & Min-K$_{20\%}$ (P) & all & $97.00\%$ \\
                   & Min-K++ (P) & all & $82.00\%$ \\
    Pythia/EZ & EZ-MIA (P) & all & $63.00\%$ \\
    Pythia/nbr. & neighbor (P) & all & $58.00\%$ \\
    \addlinespace[1pt]
    Qwen3.5 LoRA & loss (P) & $J\!<\!0.80$ & $37.00\%$ \\
    \addlinespace[1pt]
    Pythia/dedup-FT & loss (P) & $J\!\in\![0.50,0.70)$ & $92.00\%$ \\
                     & loss (P) & $J\!\in\![0.70,0.80)$ & $99.00\%$ \\
    \addlinespace[1pt]
    Pythia/semantic & loss (P) & all & $10.00\%$ \\
    \bottomrule
  \end{tabularx}
\end{wraptable}

\paragraph{Robustness across scores and environments.}
The devastating false-attribution effect is not an artifact of a specific model, score, or training run.
As detailed in \cref{tab:natural-robustness}, the failure persists across alternative LLM-specific auditing metrics (including Min-K$_{20\%}$~\citep{SAXHLBCZ24}, Min-K++~\citep{ZSYOKZYL25}, EZ-MIA~\citep{ISC26}, and neighborhood comparison~\citep{MMJSSB23}), where false-positive rates remain at least an order of magnitude above the nominal $5.00\%$ operating point.
It also survives transfer to a different architecture (Qwen3.5-2B LoRA) and withstands an executed Jaccard-$0.80$ deduplication and retraining pipeline.
In stark contrast, a boundary control using meaning-preserving but lexically disjoint semantic paraphrases drops the false attribution rate to $10.00\%$, twice the nominal $5.00\%$ but far below the natural lexical families' $99.70\%$.
In this single-exposure setting the vulnerability is therefore driven primarily by surface-level lexical similarity rather than by semantic overlap (\S\ref{sec:objective} shows a heavier fine-tuning regime where semantic families also shift the score).
Complete removed-band and alternative-operating-point results are deferred to Appendix~\ref{app:scaleup}.

\subsection{Replication on WikiText-103}
\label{sec:natdup}

We validate these findings on WikiText-103, fine-tuning $96$ GPT-2 shadows across $56$ targets with natural siblings and $32$ controls.
At a nominal $1.00\%$ FPR, natural siblings inflate the clean-reference false-member rate from $1.40\%$ to $81.00\%$, while no-sibling controls remain calibrated.
Crucially, a family-conditioned auditor retains near-perfect exact inference ($\adv\approx0.99$), confirming this failure is strictly false attribution rather than erasure.
With the exact passage absent, the family-present and family-absent worlds separate almost perfectly ($\adv=1.00$), demonstrating that these scores are highly reliable indicators of broader family exposure rather than exact provenance.
These studies show the composite null invalidates standard verdicts in the wild, motivating controlled interventions that isolate whether families substitute for or merely confound the exact record.

\section{Four-World Interventions Explain the Failure}
\label{sec:vision}

\subsection{Causal influence of faithful content}
\label{sec:trich}

To establish causality, we evaluate controlled same-label families on CIFAR-10 using $48$ shielded targets and $256$ target-disjoint bystanders, following standard shadow density estimation protocols~\citep{CCNSTT22}.
With exact membership randomized independently against a fixed family presence ($k{=}8$), faithful near-duplicates reduce vulnerable-target advantage from $0.55$ to $0.11$, a change of $-0.44$ (\cref{fig:trichotomy}).
In contrast, lower-fidelity reinforcement variants change advantage by only $-0.08$, while same-class density and random controls show null effects ($-0.03$ and $-0.05$).
Conflicting-label records instead increase advantage by $+0.24$ and reduce target accuracy.
Global accuracy remains stable, and the target-disjoint bystanders exhibit no measured average leakage change (Appendix~\ref{app:util}).
The intervention therefore depends strictly on reproducing target content, not merely adding records, changing class density, or increasing optimizer exposure.

\subsection{Substitution versus addition regimes}
\label{sec:verdict}

The four-world audit distinguishes standard exact-record evidence, family evidence, and the exact record's residual after conditioning on the family (\cref{tab:fourworld}).
Table entries use one pooled score-AUC estimator throughout and should not be interchanged with the per-target advantages reported elsewhere.
In classification, family presence is detectable but exact-given-family inference falls to near chance.
A family-conditioned auditor's TPR drops to its measured FPR and $\mathrm{LR}^{+}$ approaches $1.00$, with $88.00\%$ of the $24$ vulnerable targets falling below advantage $0.05$.
At a decision level (\cref{tab:vision-verdict}), the clean-reference auditor's FPR rises to $0.59$ on faithful-family non-members.
Recalibrating with the family removes this false-positive shift but cannot recover the exact bit, leaving both TPR and FPR at $0.11$.
Conflicting-label families provide a sign control by retaining strong separation, yielding $\mathrm{LR}^{+}=9.90$ under conflict-conditioned calibration and $\mathrm{LR}^{+}=5.40$ even under clean calibration.
This confirms that classification exhibits substitution, where faithful coverage renders the exact example redundant~\citep{FZ20,F20}.
Autoregressive language modeling behaves fundamentally differently.
CC-News family presence is strongly detectable, yet exact-given-family AUC remains $0.75$.
Because every exact exposure can continue reducing sequence loss, the exact record retains a residual.
The practical failure is nevertheless severe because the natural family alone crosses a boundary calibrated without families.
Thus, text models demonstrate additive confounding rather than classification-style erasure.

\begin{figure}[tbp]
  \centering
  \begin{minipage}[b]{0.49\linewidth}
    \centering
    \includegraphics[width=\linewidth,keepaspectratio]{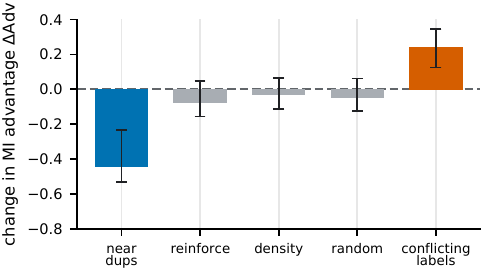}
  \end{minipage}\hfill
  \begin{minipage}[b]{0.49\linewidth}
    \centering
    \footnotesize
    \renewcommand{\arraystretch}{1.18}
    \setlength{\tabcolsep}{2.5pt}
    \begin{tabularx}{\linewidth}[b]{@{}Xrrr@{}}
      \toprule
      Calibration $\to$ evaluation & TPR & FPR & $\mathrm{LR}^{+}$ \\
      \midrule
      clean $\to$ clean & $0.37$ & $0.07$ & $5.20$ \\
      \addlinespace[1pt]
      copy-cal. $\to$ copies & $0.11$ & $0.11$ & $1.00$ \\
      clean $\to$ copies & $0.61$ & $0.59$ & $1.00$ \\
      \addlinespace[1pt]
      conflict-cal. $\to$ conflict & $0.56$ & $0.06$ & $9.90$ \\
      clean $\to$ conflict & $0.13$ & $0.02$ & $5.40$ \\
      \bottomrule
    \end{tabularx}
    \vspace*{2mm}
  \end{minipage}
  
  \begin{minipage}[t]{0.49\linewidth}
    \captionof{figure}{CIFAR-10 trichotomy. Near duplicates reduce advantage; lower-fidelity and disjoint controls remain null ($95\%$ CIs).}
    \label{fig:trichotomy}
  \end{minipage}\hfill
  \begin{minipage}[t]{0.49\linewidth}
    \captionof{table}{CIFAR-10 verdicts at nominal $5\%$ FPR; $\mathrm{LR}^{+}$: bootstrap medians. Copy/conflict calibration reproduces the evaluated family.}
    \label{tab:vision-verdict}
  \end{minipage}
  \vspace{-6mm}
\end{figure}

\begin{wraptable}{R}{0.52\linewidth}
  \centering
  \footnotesize
  \renewcommand{\arraystretch}{1.10}
  \setlength{\tabcolsep}{2pt}
  \vspace{-4mm}
  \caption{Complete four-world identification results (pooled score AUC). Columns are $E\mid F{=}0$: $H_{10}$ vs.\ $H_{00}$; $F\mid E{=}0$: $H_{01}$ vs.\ $H_{00}$; and $E\mid F{=}1$: $H_{11}$ vs.\ $H_{01}$. Chance is $0.50$; cls./AR denote classification/autoregressive modeling.}
  \label{tab:fourworld}
  \begin{tabularx}{\linewidth}{@{}>{\raggedright\arraybackslash}Xrrr@{}}
    \toprule
    Setting & Standard & Family & Residual \\
    \midrule
    CIFAR-10 cls. & $0.56$ & $0.69$ & $0.52$ \\
    CC-News AR & $0.87$ & $0.98$ & $0.75$ \\
    \addlinespace[1pt]
    AG-News cls. & $0.61$ & $0.62$ & $0.53$ \\
    AG-News AR & $1.00$ & $0.89$ & $1.00$ \\
    \bottomrule
  \end{tabularx}
    \vspace{-4mm}
\end{wraptable}

\subsection{The learning objective drives the split}
\label{sec:objective}

Because vision and text models differ in modality and architecture, we isolate the impact of the learning objective with a matched intervention on AG-News.
We hold the GPT-2 backbone, target records, family definitions (semantic paraphrase families sharing almost no $5$-grams with their targets; Appendix~\ref{app:semantic}), exposure levels, and four-world audit design strictly constant, varying only the training task between sequence classification and autoregressive modeling.
This single shift causes the exact-given-family AUC to jump from $0.53$ under classification to $1.00$ under autoregressive modeling, with the standardized surviving exact fraction rising from $0.27$ to $0.86$.
Because family presence remains independently detectable under both objectives (AUC $0.62$ and $0.89$), and these lexically disjoint families show that even purely semantic overlap shifts scores under multi-epoch fine-tuning, the divergence cannot be explained by the family simply becoming unlearnable.
Instead, this intervention definitively identifies the learning objective as the causal driver of the substitution-versus-addition split.
Finally, replacing natural CC-News families with meaning-preserving paraphrases (Jaccard $0.01$) reduces false attribution to $10.00\%$ while leaving exact-record inference fully intact (AUC $1.00$, \cref{tab:natural-robustness}), confirming that lexical near-duplication is the dominant confounding channel at single-exposure scale, though heavier fine-tuning extends it to semantic families.

\subsection{Score saturation and local invariance}
\label{sec:mechanism}
\label{sec:fidelity}
\label{sec:budget}
\label{sec:adaptive}

To formalize why the local score basin becomes invariant to exact membership, let $S$ denote the trained family of $x$, and let $\phi(x;D)$ be the expected audit score of query $x$ after training on $D$.
The exact record's discrete marginal contribution in the context of its family is
\begin{equation}
  \Delta_x(S)=\phi(x;D\cup S\cup\{x\})-\phi(x;D\cup S).
  \label{eq:marginal}
\end{equation}
Family records raise the queried target's score more strongly when the target is OUT than when it is IN, shrinking the standardized IN/OUT gap from $1.34$ to $0.10$ for near-duplicates and $0.04$ for exact copies.
Because highly faithful records cover more of the same local signal, their benefit saturates before $x$ is added, forcing $\Delta_x(S)$ to shrink.
This saturation extends across the target's entire neighborhood: querying fixed augmentations of each target reveals their collective IN/OUT gap collapses from $2.38$ without a family to $0.26$ with near-duplicates (\cref{fig:f11}).
The effect strictly follows fidelity (\cref{fig:f2}): exact copies erase the signal most strongly, while interpolating a family record toward an unrelated image smoothly removes the protection, explaining the density null (\cref{fig:f10}).
This erasure remains tightly localized to the target: the two nearest same-class training neighbors see their membership advantage change by only $0.01$ (with a confidence interval of $[-0.02,0.05]$).
This interpretation clarifies why low exact-record leakage does not imply low content leakage: the score shift is carried by $S$ and remains robustly observable during a family-presence test.

This mechanistic effect scales with the family budget $k$, appearing even with a single near-duplicate and saturating near $k=8$ (\cref{fig:f7}).
The substitution dynamics consistently transfer across different datasets and architectures (driving a $-0.76$ advantage reduction on CIFAR-100 and a $-0.30$ reduction on VGG-11), while the additive residual persists across controlled language models (\cref{fig:f3}).
Furthermore, the classification erasure regime withstands adaptive, family-aware auditors (\cref{fig:f4}).
When reference shadows are explicitly trained with the family present, advanced attacks (including full-covariance LiRA, augmentation differencing, explicit family-record queries, and RMIA) all remain near chance (AUC $0.47$--$0.53$; Appendix~\ref{app:adaptive}).
While this constitutes empirical robustness across black-box channels rather than an information-theoretic guarantee, it confirms a unified framework: fidelity dictates how strongly the family moves the target score, while the learning objective dictates whether that movement substitutes for or adds to the exact-record evidence.

\section{Model Evidence Relocates to the Content Family}
\label{sec:contentpresence}
\label{sec:baselines}

\subsection{High detectability of family presence}

Near-duplicate families redistribute rather than universally erase model evidence: what is lost in exact-record specificity is relocated to the broader content family.
With the exact target absent, family presence remains highly detectable (AUC $0.69$ in vision, $0.98$ in text).
Consequently, the very score shift that triggers a false-positive verdict for exact membership operates as a true-positive signal of family exposure.
Our controlled cluster-toggle experiment (\cref{fig:frontier}) makes this evidence relocation explicit: as family size $k$ sweeps from $1$ to $16$ in vision, exact-given-family advantage falls from $0.09$ to $0.01$, while family-presence advantage climbs from $0.12$ to $0.48$.
The isolated pythia-2.8B reference point illustrates the cross-objective contrast, preserving a detectable exact-sequence residual while simultaneously exposing the content family (AUC $0.98$).
Near-duplicate exposure thus shifts the identifiable unit of a model-only audit from the discrete record to its surrounding content family, without any change to the underlying scores.

\subsection{Limitations of composite-null calibration}
\label{sec:composite}

\begin{wrapfigure}{R}{0.52\linewidth}
  \centering
  \vspace{-4mm}
  \includegraphics[width=\linewidth]{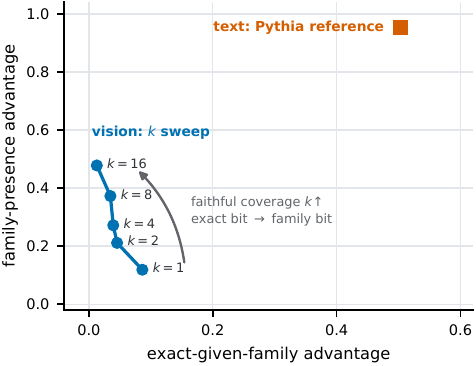}
  \caption{Information relocation. The blue path connects the five evaluated vision family sizes $k\in\{1,2,4,8,16\}$; the Pythia point is a single cross-objective reference, not a text sweep. Crossed-bootstrap intervals are reported in the appendix.}
  \vspace{-4mm}
  \label{fig:frontier}
\end{wrapfigure}
A natural defense is to calibrate the auditor against both components of the composite null hypothesis in \cref{eq:composite}.
While a worst-case threshold can control the false-positive rate across $H_{00}$ and $H_{01}$, this conservative correction cannot recreate the separation erased by the family, decimating power when the family-induced score shift is large.
In our vision experiments, the resulting composite exact test achieves only $0.53$ AUC, with statistical power matching the nominal $5.00\%$ false-positive rate (equivalent to random chance). Even an oracle four-world generalized likelihood-ratio (GLR) test reaches only $0.55$ AUC.
On CC-News text, the corresponding composite and oracle AUCs are $0.69$ and $0.74$.
Here, shifting from clean calibration ($H_{00}$) toward worst-case family calibration ($H_{01}$) reduces the worst-case false-positive rate from $0.88$ to $0.05$, but at an unacceptable cost: power on standard exact members ($H_{10}$) drops from $0.34$ to $0.00$, while power on family-present members ($H_{11}$) plummets from $0.99$ to $0.17$.
Crucially, this limitation is not an artifact of a simplistic monotone threshold: the oracle GLR test directly models all four worlds using perfect knowledge of the score distributions, yet remains bounded by these same low AUCs (Appendix~\ref{app:glrt}).
In practice, incomplete family knowledge presents a further, serious deployment obstacle.
If a text auditor identifies and accounts for only $1$ or $2$ members of an $8$-sibling family, the resulting false-positive rate remains highly elevated at $0.70$ or $0.19$, respectively.
The auditor recovers nominal calibration ($\alpha = 0.05$) only once at least half of the family ($4$ of $8$ siblings, FPR $0.06$) is reproduced and included in the reference shadows (Appendix~\ref{app:recall}).
Because exhaustively enumerating every syndicated variant, mirrored page, or minor edit across web-scale data is vastly more difficult than in this controlled best-case evaluation, robust composite-null calibration remains practically unfeasible.

\subsection{Scoped reporting and recommended abstention}
\label{sec:abstain}

When near-duplicate variants of a queried record cannot be definitively excluded from the training corpus, model-only evidence should never be treated as dispositive for proving exact-record data provenance.
Instead, we argue that a rigorous model-only audit must explicitly declare its unit of inference, distinguishing a discrete record from a broader lexical family.
To evaluate an exact-record claim, the auditor should construct both clean ($H_{00}$) and family-present ($H_{01}$) reference null distributions, disclosing separate false-positive and power rates across all four worlds ($H_{00}, H_{01}, H_{10}, H_{11}$).
If plausible near-duplicates cannot be exhaustively enumerated, or if the family-present null crosses the clean threshold, the only defensible outcome is a proactive abstention on exact membership, optionally accompanied by a positive family-presence finding.
This reporting rule preserves what the model actually reveals without silently converting general content exposure into provenance evidence.
When model-only audits must abstain, definitive exact provenance must rely on external records, such as training-set inspection, signed manifests, or system access logs.

This paradigm shift reflects the empirical boundaries of black-box audits.
Across nine tested scoring channels, including loss, logit, covariance, and neighborhood statistics, faithful classification families completely erase exact-record separation, whereas autoregressive text models preserve a detectable exact-sequence residual.
Crucially, this residual signal does not repair standard audits because family presence still falsely crosses the clean threshold with high probability.
The saturation mechanism is moreover dual use: the same data-side additions that reduce a record's exact-membership leakage can equally be used to weaken or fabricate model-only evidence, further underscoring why such verdicts are forensically invalid on web-scale corpora where natural, un-enumerated near-duplicates abound.

\section{Conclusion}
\label{sec:conclusion}

This paper exposes a fundamental vulnerability in using membership inference as a forensic tool for data provenance.
By showing that standard audits conflate exact-record inclusion with broader content-family exposure, we demonstrate that a positive model-only score is an invalid proof of individual copyright or licensing claims.
The four-world audit framework provides a rigorous methodology to separate these overlapping signals, mapping the precise boundaries of exact-record survival and erasure.
Ultimately, these findings shift the conceptual baseline of machine learning audits from treating neural networks as discrete databases toward understanding them as continuous statistical representations of broad data neighborhoods.
In doing so, this work establishes a new, defensible standard for data-provenance auditing, ensuring that future privacy and legal evaluations remain grounded in what models actually reveal.

\bibliographystyle{iclr2027_conference}
\bibliography{normal_generated_py3}

\appendix

\section{The evidence model}
\label{app:model}
\label{sec:model}

\paragraph{Terminology.} Throughout the appendix, a \emph{shield} is a record of the controlled family added around a
target (\S\ref{sec:vision}), and a \emph{shielded} target is one whose family is present in training.

The trichotomy, its fidelity ordering, and the saturation are not independent findings; one scalar model
\emph{organizes} them (the protective and null arms follow from a single concavity assumption; the backfire
arm needs one added assumption, below). Fix target $x$ of class $y$, and assign each training example $z$ an \emph{evidence weight}
toward $(x,y)$, $s(x){=}1$; a same-label example carries $s(z){=}\rho(z)\in[0,1]$, where $\rho$ is its content
fidelity to $x$ (copies $\rho{=}1$, near-duplicates slightly less, content-disjoint $\rho{\approx}0$); a
conflicting-label example near $x$ carries $s(z){=}-\beta\rho(z)$ with $\beta{>}0$. In text there is no label
channel, so a shared-prefix/divergent-suffix shield has small $s\ge 0$ and $\beta{=}0$. With total evidence
$t(D){=}\sum_{z\in D}s(z)$, model the score $\phi(x){=}g(t(D))+\varepsilon$, $\varepsilon\sim\mathcal
N(0,\sigma^2)$ over training randomness, $g$ increasing.

\begin{assumption}[saturating fit]\label{as:a1}
$g$ is strictly concave and bounded on $[0,\infty)$ with $g'\!\to\!0$.
\end{assumption}
This is cheap to defend, cross-entropy gradients on an example scale with $1{-}p_y$ and vanish as it is fit,
so each added faithful copy moves the score less; it holds exactly for faithful (positive-evidence) shields in
the kernel instantiation (Appendix~\ref{app:krr}). The backfire arm alone needs one more
ingredient.
\begin{assumption}[boundary steepness]\label{as:a2}
Extended to conflicting evidence $t<0$, $g$ is S-shaped; $g'$ is unimodal with mode $t^\star\le0$, so $g$ is
convex below $t^\star$ and concave above (A\ref{as:a1} is its restriction to $t\ge0$).
\end{assumption}
The protective and null arms use A\ref{as:a1} alone; A\ref{as:a2} is invoked only for conflict.

\begin{proposition}[gap identity, mean level]\label{prop:gap}
With shield evidence $t_0$ present in both worlds, $\mu_{\mathrm{in}}-\mu_{\mathrm{out}}=g(t_0{+}1)-g(t_0)$.
Every shield effect on the exact-record mean gap is the discrete derivative of $g$ at $t_0$.
\end{proposition}
\begin{corollary}[advantage readout]\label{cor:adv}
For Gaussian scores the per-record advantage is a function of $\Delta\mu/\sigma_p$, where
$\sigma_p{=}\sqrt{\sigma_{\mathrm{in}}^2+\sigma_{\mathrm{out}}^2}$ is the pooled dispersion (AUC$=\Phi(\Delta\mu/\sigma_p)$);
so \emph{at pooled dispersion held fixed across the budget}, it is strictly increasing in the mean gap $\Delta\mu$.
\end{corollary}
Concavity forces $\mu_{\mathrm{out}}$ to rise more than $\mu_{\mathrm{in}}$ ($g(t_0){-}g(0)>g(t_0{+}1){-}g(1)$),
the asymmetry of Figure~\ref{fig:f11}a ($+5.45$ vs $+3.10$).

\begin{proposition}[the four arms]\label{prop:tri}
Let $h(t){=}g(t{+}1)-g(t)$; under A\ref{as:a1} it is positive, strictly decreasing, and $h(t)\!\to\!0$.
Then
(i) in the \emph{protect} arm, faithful shields give $t_0{=}k\rho$, so the gap $h(k\rho)$ strictly decreases in $k$ and
$\rho$ (Fig.~\ref{fig:f10}); (ii) in the \emph{null} arm, content-disjoint additions have
$\rho{\approx}0$, so $t_0{=}0$ and the gap equals baseline regardless of count (density, random); (iii) in the
\emph{backfire} arm, under A\ref{as:a2}, a conflict count $m$ giving $t_0{=}{-}\beta m$ with $t^\star{\le}t_0{<}0$ yields
$h(t_0){>}h(0)$, widening the gap (non-monotone in $m$); (iv) in \emph{text}, $\beta{=}0$ and
first-token-only conflict give $t_0{\approx}0^+$, hence null-to-slightly-protective.
\end{proposition}
The proofs are one line ($h'{=}g'(t{+}1)-g'(t){<}0$ by strict concavity; (iii) uses A\ref{as:a2}, i.e.\ $g'$
peaking at $t^\star{\le}0$; Appendix~\ref{app:krr}). Thus A\ref{as:a1} \emph{alone} forces the protect and null
arms and the fidelity ordering, the backfire arm additionally requires A\ref{as:a2}, and the text null
($-0.09,-0.00$) follows from $\beta{=}0$; the model \emph{organizes} the trichotomy rather than deriving all
four arms from one assumption.

\paragraph{Tested prediction.} A\ref{as:a2} predicts the backfire is non-monotone in the conflict budget, the
\emph{advantage} should eventually turn over. It does not; sweeping $k_{\mathrm{conf}}$ to $256$
(Figure~\ref{fig:f13}), $\Delta\adv$ plateaus at ${\sim}0.30$ even at $8\times$ the budget. This is a failure of
\emph{Corollary~\ref{cor:adv}'s budget-invariant-dispersion assumption}, not of the mean-gap model; the moment
decomposition and its caveats are given once, in Appendix~\ref{app:krr} (Limits). We report this plainly rather
than rescue it, and the backfire \emph{sign} needs no A\ref{as:a2} in any case; the $\mu$-asymmetry
($\mu_{\text{out}}$ crashes more than $\mu_{\text{in}}$, \S\ref{sec:mechanism}) already delivers it. The model's
\emph{positive} record stands. It unifies the four arms under one decreasing $h$, quantitatively fits the budget
curve ($\lambda'{\approx}3.5$, out-of-sample to $k{=}16$; \S\ref{sec:budget}), reproduces the Truth-Serum
amplification at $k{=}8$ ($+0.24^*$, \S\ref{sec:truthserum}), and its basin analysis (Prop.~\ref{prop:basin},
\S\ref{sec:adaptive}) located the residual score-channel surface that the full-covariance and differencing
attacks then probed and found empty.

\begin{figure}[t]
\centering
\begin{subfigure}[t]{0.48\columnwidth}
  \centering
  \includegraphics[width=\linewidth]{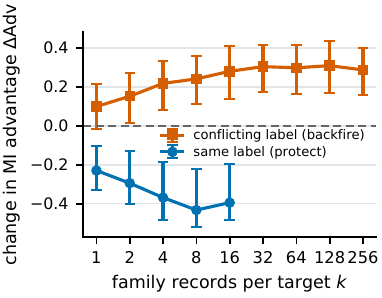}
  \caption{Advantage change vs. budget}
\end{subfigure}\hfill
\begin{subfigure}[t]{0.48\columnwidth}
  \centering
  \includegraphics[width=\linewidth]{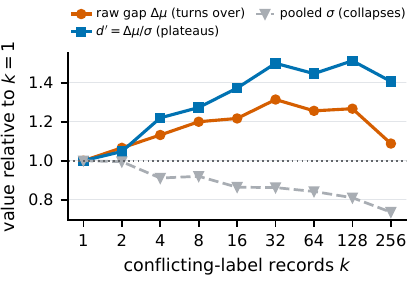}
  \caption{Moment decomposition}
\end{subfigure}
\caption{Pushing the conflicting budget to $k{=}256$. \textbf{(a)} Two arms of one decreasing-$h$ mechanism.
Same-label evidence saturates \emph{protection}, conflicting-label evidence saturates \emph{backfire}
(Prop.~\ref{prop:tri}); the backfire $\Delta\adv$ plateaus at ${\sim}0.30$ with no turnover, and its
$k{=}8$ point matches the Truth-Serum amplification (\S\ref{sec:truthserum}). \textbf{(b)} The
A\ref{as:a2}-predicted turnover does not appear in the advantage; the moment decomposition and its caveats are
in Appendix~\ref{app:krr}. Panel (b) series are normalized to $k{=}1$. Note $k{=}256$ is a
\emph{mechanism probe} (a ${\sim}25\%$ dataset increase), not a deployment budget; the protective sweep stays
at $k{\le}16$ (sub-percent). Vulnerable targets, crossed-bootstrap 95\% CIs.}
\label{fig:f13}
\end{figure}

\section{Kernel instantiation and proofs}
\label{app:krr}

We instantiate the evidence model of \S\ref{sec:model} with kernel ridge regression (equivalently a Gaussian
process posterior mean), where Assumption~\ref{as:a1} and the \emph{protective} statements
(Prop.~\ref{prop:gap}, Prop.~\ref{prop:tri}(i,ii)) hold \emph{exactly}; the backfire arm (A\ref{as:a2}) and the
basin bound (Prop.~\ref{prop:basin}) need the added assumptions noted below and are \emph{not} instantiated
here.

\paragraph{Setup.} Predict the score for $x$ by KRR against a background value $b$ with ridge $\lambda$, and
suppose $x$'s evidence is carried by a cluster of $n$ points each at kernel similarity $\kappa$ to $x$,
against an uninformative background. The posterior mean at $x$ is
\begin{equation}
g(n)=b+\Big(1-\tfrac{\lambda}{n\kappa+\lambda}\Big)(y-b), \qquad g(0)=b,
\label{eq:gn}
\end{equation}
which is strictly increasing in $n$ toward the label $y$. This is exact for $n$ \emph{identical} copies of $x$ (rank-one Gram
$K{=}\kappa\mathbf 1\mathbf 1^\top$, each at cross-similarity $\kappa$ to $x$); for distinct near-duplicates the
Gram is only approximately rank-one and \eqref{eq:gn} is a scalar-evidence heuristic that treats the cluster's
evidence as additive ($n\kappa$). We use it only for the protective (positive-evidence) arms.

\paragraph{Concavity (A\ref{as:a1}).} With $n$ continuous, $g'(n)=(y{-}b)\lambda\kappa/(n\kappa+\lambda)^2>0$
and $g''(n)=-2(y{-}b)\lambda\kappa^2/(n\kappa+\lambda)^3<0$ for $y>b$, and $g'(n)\to0$. Hence $g$ is strictly
concave and bounded, so A\ref{as:a1} holds exactly. Note $g'$ is \emph{monotone} on the valid domain
$n\ge0$. This $g$ realizes the protective and null arms exactly but does \emph{not} instantiate the backfire
arm, which requires the S-shaped $g$ of A\ref{as:a2} (an interior $g'$ peak at $t^\star\le0$), a separate
modeling assumption the rational KRR $g$ does not satisfy.

\paragraph{Gap and budget curve (Prop.~\ref{prop:gap}, \ref{prop:tri}).} The IN world has the original plus
$k$ copies ($n{=}k{+}1$), the OUT world only the $k$ copies ($n{=}k$), so
\begin{align}
\Delta(k)&=g(k{+}1)-g(k)=\frac{(y{-}b)\,\kappa\lambda}{(k\kappa+\lambda)\big((k{+}1)\kappa+\lambda\big)},\nonumber\\
\frac{\Delta(k)}{\Delta(0)}&=\frac{\lambda'(1{+}\lambda')}{(k{+}\lambda')(k{+}1{+}\lambda')},\qquad \lambda'{=}\lambda/\kappa,
\label{eq:dk}
\end{align}
a one-parameter curve, strictly decreasing in $k$, $\to0$
(Prop.~\ref{prop:tri}(i); the fit in Fig.~\ref{fig:f7}). Fidelity enters as $\kappa$. A shield of similarity
$\kappa_s<\kappa$ contributes $k\kappa_s/\kappa$ effective copies, so protection is monotone in fidelity;
content-disjoint shields have $\kappa_s\approx0$ and leave $\Delta$ at baseline (the null). With
$\phi(x)=g(\cdot)+\varepsilon$, for Gaussian IN/OUT scores per-record
$\mathrm{AUC}=\Phi(\Delta\mu/\sigma_p)$ \emph{exactly} (pooled $\sigma_p{=}\sqrt{\sigma_{\text{in}}^2+\sigma_{\text{out}}^2}$;
no equal-IN/OUT-variance assumption is needed). This map is only \emph{qualitative} in practice. Our measured
pairs depart from a single-$\sigma$ reading by up to ${\sim}2\times$ (mild non-Gaussianity and budget-varying
$\sigma_p$; unequal $\sigma_{\text{in}},\sigma_{\text{out}}$ does \emph{not} affect it), so we treat
Cor.~\ref{cor:adv}'s advantage claim as monotone in $\Delta\mu/\sigma_p$, not as a specific functional form.

\paragraph{Basin invariance (Prop.~\ref{prop:basin}).} For a query $q$ at similarity $\kappa_q$ to $x$, the
IN$-$OUT shift is $\delta(q)\propto(y{-}b)\kappa_q\lambda/(n_q\kappa+\lambda)^2$, where $n_q$ counts nearby
evidence. Where $\kappa_q$ is large (near $x$), faithful shields make $n_q$ large, so $\delta(q)\to0$ by the
same saturation; where $\kappa_q\approx0$ (far), $\delta(q)\approx0$ directly. The Neyman--Pearson bound for
a Gaussian location family with \emph{isotropic}, independent-per-query noise ($\Sigma{=}\sigma^2 I$) and mean
gap $\lVert\boldsymbol\delta\rVert_2$ gives Prop.~\ref{prop:basin}; correlated queries instead require the
Mahalanobis $\sqrt{\boldsymbol\delta^\top\Sigma^{-1}\boldsymbol\delta}$ and fall outside it (handled
empirically, \S\ref{sec:adaptive}). For exact copies ($\kappa_s{=}\kappa$), $\delta(q)\to0$ for all $q$ as
$k\to\infty$.

\begin{remark}[mean-shift decomposition]\label{cor:cons}
\begin{align*}
& \underbrace{g(t_0{+}1)-g(t_0)}_{\text{membership}\,\mid\,\text{cluster}}
 \;+\; \underbrace{g(t_0)-g(0)}_{\text{cluster presence}} \\[-1pt]
&\qquad =\; g(t_0{+}1)-g(0)\ \ge\ g(1)-g(0)\ \ (\text{undefended}).
\end{align*}
\end{remark}
\paragraph{Mean-shift decomposition (Rem.~\ref{cor:cons}).} Equation~\eqref{eq:gn} telescopes directly; the
total is a lower bound, \emph{not} a conserved quantity (it grows with $k$), and as $k$ grows the
point-membership term $\to0$ while the cluster-presence term grows.

\paragraph{Limits.} The model assumes (a) a scalar sufficient statistic; (b) a pooled dispersion $\sigma_p$
\emph{constant across the budget}, violated in practice ($\sigma_p$ shrinks as $k$ grows, i.e.\
heteroscedasticity across budgets, which is what erases the predicted advantage turnover in the confusion
sweep); (c) an identical-copy rank-one Gram matrix for the exact
derivation, whose extension to distinct near-duplicates is a scalar-evidence heuristic, not an established bound; (d) A\ref{as:a2}'s S-shape for the backfire arm, which the rational KRR $g$ does not satisfy;
and (e) isotropic query noise for Prop.~\ref{prop:basin}, so correlated-query attacks are covered only
empirically (\S\ref{sec:adaptive}). A\ref{as:a1} is proved here for the kernel case and assumed for deep
networks (cross-entropy gradient decay). It is a \emph{mechanism} model consistent with the phenomenology, not
a security proof, and speaks only to the score channel; non-score attacks (training-set access,
near-duplicate counting) fall outside the query-only access considered here. Its one prediction beyond the
fitted phenomenology, a conflicting-budget turnover in the \emph{advantage} (Prop.~\ref{prop:tri}(iii)), is
\emph{not} borne out; the advantage plateaus out to $k{=}256$ (Fig.~\ref{fig:f13}). Only a post-hoc
decomposition suggests a turnover in the unobservable mean gap ($\Delta\mu$ peaks near $k{\approx}32$ at $3.64$
and declines to $3.01$ by $k{=}256$; paired $+0.63$, but a single post-selected contrast on a target-only
bootstrap, so suggestive not confirmatory), while the pooled $\sigma$ falls ${\approx}15\%$ over $k{=}32{\to}256$
(${-}26\%$ over $k{=}1{\to}256$), leaving $d'{=}\Delta\mu/\sigma$, and the advantage, flat. Separately, per-instance DP already shows that records with similar
neighbors enjoy tighter per-instance bounds than isolated outliers, which is precisely what a shield
manufactures, so the connection to a formal per-record guarantee has analytic support, not merely conjecture;
making it quantitative for shielded records is the closest route from this heuristic toward a certificate. The
$\sigma$ collapse itself has a mechanistic reading. Heavy conflicting injection makes the local fit
near-deterministic (the model confidently predicts the wrong label regardless of $x$'s membership), so
$\phi(x)$ becomes less sensitive to the membership mask.

\paragraph{Reproducibility details.} We fit $\lambda'$ in $d'$-space over $k\!\in\!\{1,2,4,8\}$ through the
Gaussian AUC relation ($k{=}16$ treated as saturation noise), yielding $\lambda'{\approx}3.5$, $R^2{=}0.74$; the
curve's apparent floor at large $k$ is the AUC map saturating, so the fit does not pass through $k{=}1$ exactly
($-0.18$ predicted vs the measured $-0.23$). Shadow-model counts per experiment follow.
\begin{table}[tbp]
\centering\small
\caption{Number of shadow models used in each experiment. Counts separated by a slash report the baseline and intervention arms, respectively.}
\label{tab:shadow-counts}
\begin{tabular}{lc}
\toprule
experiment & shadows \\
\midrule
baseline & $96$ \\
main trichotomy arms & $48$ \\
generalization (CIFAR-100, VGG; baseline/arms) & $64/32$ \\
text (scratch, finetune) & $32$ \\
adaptive LiRA, RMIA, mechanism & $96$ \\
DP-SGD $\epsilon$-sweep, GN control & $24$ each \\
De-style DP & $16$ \\
$k$-sweep, interp, confusion budget & $48$ \\
why-keep-$x$ (utility; trade-off probe) & $64$ ($32$) \\
CC-News scale-up (\emph{none}/\emph{natdup}) & $64/64$ \\
WikiText natural siblings & $96$ \\
WikiText erasure frontier (per arm) & $96$ \\
AG-News objective control (cls./AR) & $64/64$ \\
CC-News semantic \emph{natdup} & $28$ \\
dedup-and-retrain (\emph{natdup}) & $64$ \\
Qwen3.5 LoRA ablation & $48$ \\
\bottomrule
\end{tabular}
\end{table}

\section{Adaptive-attack details}
\label{app:adaptive}

This section gives the full adaptive-attack analysis summarized in \S\ref{sec:adaptive}. We build a
shield-aware LiRA, a per-target diagonal-Gaussian likelihood ratio over the joint vector
$z=[\phi(x), \phi(\text{$16$ augmentations of }x)]$, calibrated on \emph{discovery} shadows that themselves
contain the shields, frozen, and evaluated on disjoint shadows. This attacker is stronger than naive LiRA on
undefended records. On the vulnerable targets its TPR@$1\%$FPR is $0.24$ versus naive $0.17$
(Figure~\ref{fig:f4}), so the comparison is not against a weakened adversary. On shielded records its
TPR@$1\%$FPR falls to $0.01$ (neardup) and $0.00$ (copies), with AUC $\approx 0.50$--$0.53$, at the $1\%$-FPR
floor and a ${>}20\times$ reduction from the adaptive baseline. (All adaptive-attack numbers are on the same
vulnerable-$24$ targets as RMIA below, for a like-for-like comparison.) We do not claim immunity to all future
attacks; we report that this stronger, shield-aware attacker gains essentially no excess true-positive rate
over the $1\%$-FPR floor that matters for privacy.

\paragraph{Attackers the diagonal LR cannot represent.} A diagonal likelihood ratio cannot capture
cross-augmentation correlation or the natural shield-aware statistic $\phi(x)$ minus the local basin level.
We therefore add two stronger attacks on the same features, namely a \emph{full-covariance} joint LiRA over $z$
(isotropic shrinkage, intensity $0.3$, for the $24$-sample/$17$-dim estimator), and the explicit difference
$\phi(x)-\overline{\phi(\text{aug})}$ (the text-domain analogue is the neighbourhood-comparison attack of
\citet{MMJSSB23}). On undefended
records the full-covariance attack matches the diagonal one (AUC $0.84$ vs $0.86$), confirming it is a real
attacker; the differencing statistic is weak \emph{even there} (AUC $0.53$), because subtracting the shared basin
removes the membership signal along with the basin. On shielded records both collapse to chance
(full-covariance AUC $0.49$--$0.51$, difference AUC $0.47$--$0.50$). To rule out that the $17$-dimensional
covariance is merely under-determined (${\sim}24$ samples), we also run the optimal contrast as a
\emph{well-conditioned $2$-dimensional} full-covariance LR on $[\phi(x),\overline{\phi(\text{aug})}]$. It is
strong undefended (AUC $0.84$, TPR@$1\%$ $0.22$) yet also floors on shielded records ($0.53/0.01$ neardup,
$0.51/0.00$ copies), and sweeping the $17$-d shrinkage over $\{0.1,0.3,0.5\}$ leaves the shielded AUC in
$[0.49,0.52]$ throughout. Even at $x$ itself, where residual curvature might let signal
survive, none does. The exact-record signal remains at chance against every LiRA variant we construct,
diagonal or full-covariance, aggregating or differencing.

\begin{figure}[tbp]
\centering
\includegraphics[width=0.5\columnwidth]{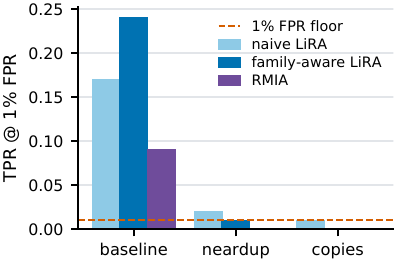}
\caption{A shield-aware adaptive attacker is stronger than naive LiRA on undefended records (baseline) yet
is driven to the $1\%$-FPR floor on shielded records; RMIA, a structurally different attack family, behaves
likewise.}
\label{fig:f4}
\end{figure}

\paragraph{An attacker that reads the shields.} The shields are treated as \emph{non-secret}, so the
strongest adaptive attacker should query the \emph{actual} $k{=}8$ shield images, not generic augmentations of
$x$. We record $\phi$ on the real shields per shadow and fit a per-target full-covariance LR over
$[\phi(x),\phi(\text{shield}_1),\dots,\phi(\text{shield}_8)]$. It gains \emph{nothing}. On the shielded world
its AUC is $0.52$ (TPR@$1\%$ $0.01$), indistinguishable from naive LiRA's $0.52$, and combining the shields
with the $16$ augmentations ($25$-dim) also floors ($0.51$). Reading the known shields does not help because
they are present in \emph{both} the member and non-member worlds and near-saturated, which makes them membership-invariant,
exactly the basin of Prop.~\ref{prop:basin}. This closes the most natural adaptive gap. Even an attacker who
knows and queries every shield floors.

\paragraph{Why score-channel attacks floor.} The model of \S\ref{sec:model} explains this under one stated
assumption. Model any black-box attacker as choosing queries $q_1,\dots,q_m$ and observing
$\phi(q_j){=}g(t_j+\rho_j\,\mathbf 1[x\in D])+\varepsilon_j$, where $t_j$ is shield evidence at $q_j$ and
$\rho_j$ is $x$'s own influence there. IN and OUT induce Gaussians differing only in mean, with per-query
shift $\delta_j{=}g(t_j{+}\rho_j)-g(t_j)$.
\begin{proposition}[basin invariance, independent-query case]\label{prop:basin}
If the per-query noise is independent ($\varepsilon_j$ i.i.d., $\Sigma{=}\sigma^2 I$), then for any test
$\mathrm{TPR}\le\Phi\!\big(\Phi^{-1}(\mathrm{FPR})+\lVert\boldsymbol\delta\rVert_2/\sigma\big)$. Near $x$,
$\rho_j$ is large but so is $t_j$ (faithful shields shadow $x$'s influence profile), so $\delta_j$ is killed by
saturation (A\ref{as:a1}); far from $x$, $\delta_j$ is killed because $\rho_j{\approx}0$. For exact copies the
shadowing is perfect and $\sup_j\delta_j\!\to\!0$ as $k\!\to\!\infty$.
\end{proposition}
So for (near-)independent queries, small $\boldsymbol\delta$ bounds \emph{every} such attack. This is why the
shield-aware diagonal LiRA and RMIA floor (the theorem is \emph{consistent with}, not proof of, those empirical
results). The bound does \emph{not} cover correlated queries. For those, the
Neyman--Pearson-optimal statistic uses the Mahalanobis distance
$\sqrt{\boldsymbol\delta^\top\Sigma^{-1}\boldsymbol\delta}$, which can exceed
$\lVert\boldsymbol\delta\rVert_2/\sigma$ along a low-variance contrast direction such as
$\phi(x){-}\overline{\phi(\text{aug})}$. That is precisely the residual surface, and the full-covariance and
differencing attacks above are built to exploit it; that they nonetheless collapse to chance on shielded
records is therefore \emph{empirical} evidence that $\boldsymbol\delta$ itself vanishes there, not a
consequence of the bound. For near-duplicates ($\kappa_s{<}\kappa$) this empirical robustness, not the
theorem, is the operative statement.

\paragraph{A different attack family, RMIA.} A skeptic might worry the defense merely evades LiRA's specific
statistic. We therefore also evaluate RMIA~\citep{ZLS24}, which uses a fundamentally different
test. Instead of a per-target Gaussian on the target's own loss, RMIA computes the online pairwise
likelihood ratio of the target against a \emph{population} of reference points $z$, scoring
$\Pr_{z}[\,(p_\theta(x)/\bar p(x))/(p_\theta(z)/\bar p(z)) > \gamma\,]$ (our port of the official implementation, $\gamma{=}2$,
bystanders as the population). On the same vulnerable-$24$ records RMIA is a working but \emph{weaker} attack
than LiRA on this subset (AUC $0.68$, TPR@$1\%$FPR $0.09$, versus naive LiRA's $0.80/0.17$;
Figure~\ref{fig:f4}). On shielded records it too is driven to chance, with AUC $0.48$
(neardup) and $0.50$ (copies), with TPR@$1\%$FPR $=0.00$. That a second,
structurally different attack family, even if weaker here, fails on the same shields indicates the effect
is not specific to LiRA's statistic.

\section{Secondary reading: the family intervention as exact-record camouflage}
\label{app:secondary}

This appendix collects experiments that read the controlled family intervention of \S\ref{sec:vision} from the
contributor's side. These results are \emph{not} part of our main
claims. As stated in \S\ref{sec:related}, any defensive interpretation is strictly secondary, and the mechanism is
dual use (Ethics statement). The same data-side addition that suppresses exact-record evidence for a
contributor is available to any party that can place records in a training corpus, which is precisely why
model-only verdicts should not be treated as proof of exact inclusion. We report the protective effect, its cost,
and its limits so that readers can judge how cheaply exact-record evidence can be moved, not to propose a
privacy defense. Content presence remains detectable throughout (\S\ref{sec:contentpresence}).

\subsection{Truth Serum poisoning}
\label{sec:truthserum}

Truth Serum~\citep{TSJLJHC22} is the offensive mirror of shield data. An adversary with the ability to inject a few
points inserts \emph{mislabeled replicas} of the target to \emph{amplify} its leakage (Chameleon~\citep{CSOU24}
does the same for label-only attacks via adaptive poisoning). We reproduce it
($k{=}8$ clean copies of each target with a wrong label) and confirm amplification. Vulnerable advantage
rises by $+0.24^*$; indeed our confusion arm (\S\ref{sec:trich}) is the same
mechanism, which is why it backfires. We then ask whether an added faithful family counteracts such an
attacker. Adding neardup shields on top of the Truth-Serum poison reduces the advantage by $-0.34^*$ versus
the no-poison baseline, turning the attacker's $+0.24$ amplification into net protection, close to the
unpoisoned shield's $-0.44$. Faithful, correct-label coverage counteracts the mislabeled poison, so even under
active poisoning the exact-record signal falls back close to the unpoisoned shielded level. Exact-record evidence can
thus be pushed in either direction by cheap data-side additions.

\subsection{Why keep the record? A negative result on local influence}
\label{sec:whykeep}

If faithful shields make $x$'s content persist, the natural challenge is why the real $x$ should be kept at all.
In a separate $64$-shadow run of the cluster assay (Figure~\ref{fig:f12}; hence exact-record advantages here read
$0.15$ vs.\ that figure's $0.14$), we measure the marginal value of keeping $x$ as
$\text{margin}(x{+}\text{shields})-\text{margin}(\text{shields-only})$ on $x$ and on $16$ held-out neighbors of
$x$ (fresh augmentations, never used as shields), paired over targets. At high fidelity ($k{=}8$ neardup) $x$ is
\emph{task-redundant}; keeping it changes accuracy and margin by $\Delta{\approx}0$ (n.s.), because the shields
already saturate its neighborhood, which is the substitution mechanism of \S\ref{sec:verdict} seen from the
utility side. On these local metrics, withholding $x$ and keeping its shields matches $x{+}$shields while making
exact membership definitively false ($0$ vs $0.15$). The margin gain from keeping $x$ reappears only as fidelity
falls ($+0.86^*$ at half-fidelity, $+1.23^*$ at quarter-fidelity, target-only bootstrap), and it is essentially
$x$'s counterfactual self-influence, the very quantity MI exploits, so the two axes are two views of one number.
We do not measure any downstream-task advantage from keeping $x$; where the literal record must be trained
(e.g.\ under a contractual provenance obligation verified out of band), that requirement is external to these
metrics.

\subsection{DP-SGD comparison details}
\label{app:dp}
\label{sec:dp}

We compare the shield's exact-record effect with DP-SGD on the targeted records. This is a comparison of
mechanisms, not a proposal to replace DP. Against a \emph{tuned} DP-SGD~\citep{ACGMMTZ16}
baseline (Opacus~\citep{YSSTPMNGBZCM21}, GroupNorm, augmentation, tuned at $\epsilon{=}8$, swept
$\epsilon\in\{1,\dots,64\}$, strong evaluation per; Table~\ref{tab:dp}),
\emph{no} operating point reaches the shield's corner (advantage $0.11$ at accuracy $0.90$), because a targeted
data intervention escapes the global accuracy-vs-noise tradeoff; a non-private GroupNorm control reaches advantage
$0.37$, so part of DP's effect is architectural. We do \emph{not} claim Pareto-dominance. DP gives an all-record
worst-case guarantee whereas the shield helps only its $48$ targets, and the full \citet{DBHSB22}
recipe at scale is untested. We claim only that on the targeted records the shield beats every DP-SGD point \emph{we test}.

\begin{table}[h]
\centering\small
\caption{Tuned DP-SGD vs shield on the protected (vulnerable) records, CIFAR-10/ResNet-9. Baseline
vulnerable advantage $0.55$; accuracy on the clean holdout. The non-private control (GroupNorm pipeline,
$\epsilon{=}\infty$) indicates how much of the advantage reduction is architectural versus from privacy.}
\label{tab:dp}
\begin{tabular}{lcc}
\toprule
defense & test acc. & remaining adv. ($\Delta\adv$) \\
\midrule
no defense                 & 0.90 & 0.55 \\
\textbf{shield (neardup)}  & \textbf{0.90} & \textbf{0.11} ($-0.44^*$) \\
non-private (GN, $\epsilon{=}\infty$) & 0.80 & 0.37 ($-0.18$) \\
DP-SGD $\epsilon{=}64$     & 0.68 & 0.31 ($-0.24^*$) \\
DP-SGD $\epsilon{=}16$     & 0.62 & 0.23 ($-0.32^*$) \\
DP-SGD $\epsilon{=}8$      & 0.57 & 0.21 ($-0.34^*$) \\
DP-SGD $\epsilon{=}4$      & 0.51 & 0.24 ($-0.31^*$) \\
DP-SGD $\epsilon{=}1$      & 0.30 & 0.27 ($-0.27^*$) \\
\bottomrule
\end{tabular}
\end{table}

Along the $\epsilon$ sweep, accuracy and protection trade off as they must. Larger $\epsilon$ buys accuracy
($0.57{\to}0.62{\to}0.68$ at $\epsilon{=}8,16,64$) but advantage rises back toward the non-private $0.37$
($0.21{\to}0.23{\to}0.31$; Table~\ref{tab:dp}). A truncated run of the high-accuracy \citet{DBHSB22}
recipe (wide WRN-16-4, augmentation multiplicity, $120$ epochs, a budget tractable for a shadow-model study)
reaches the same $0.57$ accuracy at $\epsilon{=}8$ with a vulnerable-target advantage no lower than the simpler
DP baseline, so it does not reach the shield's corner either; its published headline (${\sim}0.81$) needs a
hundreds-of-epochs regime that a per-record LiRA study of dozens of models cannot afford. We therefore claim
only that every DP configuration we \emph{successfully evaluated} stays off the $0.90/0.11$ corner, not that a
better architecture could not improve accuracy at fixed $\epsilon$.

\subsection{Utility and bystander effects}
\label{app:util}

Across all arms, global test accuracy is statistically unchanged
from baseline, and the target's own accuracy when OUT is preserved for every \emph{same-label} arm (e.g.,
OUT-accuracy $\geq 0.97$ in the fidelity sweep); only confusion harms it. Critically, the \emph{bystander}
(non-shielded) targets show $\Delta\adv \approx 0$ in every arm and dataset; shielding one record does not
expose others. This is a powered null, not absence of data. Over the $256$ bystander targets the neardup
effect is $\Delta\adv{=}{-}0.00$ with a $95\%$ CI half-width of $\pm0.02$, so any externality larger than
${\sim}0.02$ would have shown up on average. Unlike the onion effect's removal, additive shielding has no
measured \emph{average} externality on random bystanders; we do not test neighborhood-matched, rare, or
worst-case-tail records, where a localized defense would most plausibly act, so this bounds the average, not
the worst case.

\section{Supplementary figures}
\label{app:figs}

Figures~\ref{fig:f11} and~\ref{fig:f10} give the mechanism detail for \S\ref{sec:mechanism};
Figures~\ref{fig:f7} and~\ref{fig:f4} support the budget sweep (\S\ref{sec:budget}) and the
adaptive-attack analysis (\S\ref{sec:adaptive}).

\begin{figure}[t]
\begin{minipage}[t]{0.48\columnwidth}
  \centering
  \includegraphics[width=\linewidth]{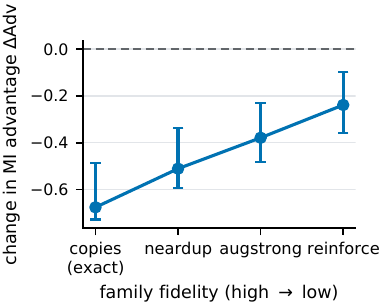}
  \caption{Controlled fidelity ordering, online augmentation disabled. Exact copies reduce exact-record
advantage most, then near duplicates, strong augmentations, and reinforcement variants; the ordering isolates
stored family fidelity from online augmentation.}
  \label{fig:f2}
\end{minipage}\hfill
\begin{minipage}[t]{0.48\columnwidth}
  \centering
  \includegraphics[width=\linewidth]{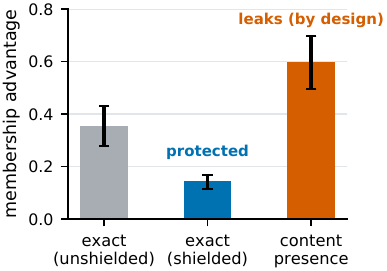}
  \caption{What the family intervention changes. Faithful families weaken \emph{exact-record} inference (advantage
$0.35\!\to\!0.14$) while \emph{content presence} stays easily detectable ($0.60$, scoring $\phi(x)$ with $x$
absent), because the family makes $x$'s content persist ($48$ shielded targets, CIFAR-10).}
  \label{fig:f12}
\end{minipage}
\end{figure}

\begin{figure}[tbp]
\centering
\begin{subfigure}[t]{0.48\columnwidth}
  \centering
  \includegraphics[width=\linewidth]{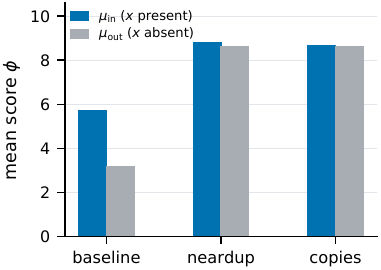}
  \caption{Score collapse}
\end{subfigure}\hfill
\begin{subfigure}[t]{0.48\columnwidth}
  \centering
  \includegraphics[width=\linewidth]{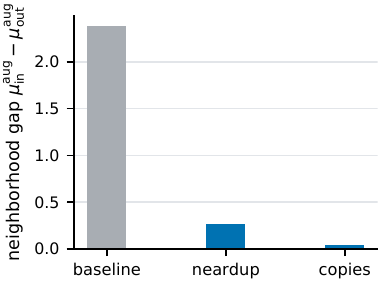}
  \caption{Neighborhood gap}
\end{subfigure}
\caption{Mechanism. (a) Shields raise the OUT-world score $\mu_{\rm out}$ to meet $\mu_{\rm in}$, collapsing
the IN/OUT gap ($d'$, $1.34\!\to\!0.10\!\to\!0.04$). (b) The target's neighborhood (its augmentations)
becomes confident and membership-invariant under shielding.}
\label{fig:f11}
\end{figure}

\begin{figure}[t]
\begin{minipage}[t]{0.48\columnwidth}
  \centering
  \includegraphics[width=\linewidth]{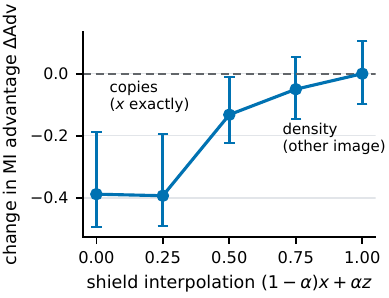}
  \caption{Causal test. As shields interpolate from the exact target ($\alpha{=}0$) toward a random same-class
image ($\alpha{=}1$), protection decays monotonically to the density null, so proximity to $x$ is the lever.}
  \label{fig:f10}
\end{minipage}\hfill
\begin{minipage}[t]{0.48\columnwidth}
  \centering
  \includegraphics[width=\linewidth]{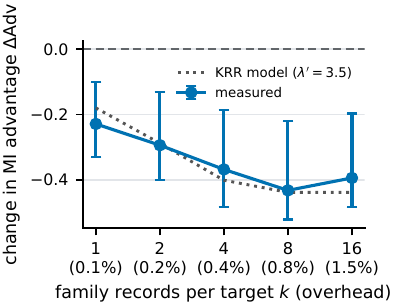}
  \caption{The exact-record advantage reduction grows with the shield count per target and saturates near
$k{=}8$; even one near-duplicate ($0.1\%$ overhead) yields a significant reduction. Vulnerable targets,
CIFAR-10/ResNet-9, $95\%$ CI.}
  \label{fig:f7}
\end{minipage}
\end{figure}

\paragraph{Calibration--power tradeoff vs.\ the assumed family prior (Fig.~\ref{fig:prior}).} The nuisance-aware
auditor of \S\ref{sec:composite} controls FPR on the composite exact-record null by assuming a family prior
$\pi{=}P(F{=}1)$ and calibrating one threshold $\tau$ on the mixture $(1{-}\pi)H_{00}+\pi H_{01}$; $\pi{=}0$ is the
clean-reference auditor and $\pi{=}1$ the minimax (worst-case-$H_{01}$) auditor of \S\ref{sec:composite}. Sweeping
$\pi$ over $[0,1]$ shows the section's two failure modes are endpoints of a strict tradeoff. In \emph{text}
(pythia-2.8b), the worst-case FPR over $\{H_{00},H_{01}\}$ falls from $0.88$ ($\pi{=}0$, false attribution) to
$0.05$ ($\pi{=}1$, controlled) exactly as power on the \emph{family-absent} exact alternative $H_{10}$ collapses
from $0.34$ to $0$ (and on $H_{11}$ from $0.99$ to $0.17$). No prior both controls the family-present null and
keeps exact-record power. In \emph{vision} power stays ${\approx}\alpha$ ($0.05$--$0.07$) for \emph{every}
$\pi$, because faithful coverage has removed the exact bit, so no calibration can recover it.

\begin{figure}[t]
\begin{minipage}[t]{0.48\columnwidth}
  \centering
  \includegraphics[width=\linewidth]{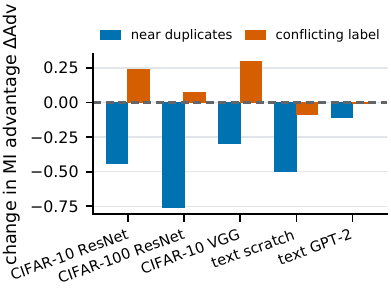}
  \caption{Faithful-coverage protection holds across datasets, architectures, and modalities (including a
finetuned $124$M GPT-2). The confusion backfire is vision-specific; in text it is null.}
  \label{fig:f3}
\end{minipage}\hfill
\begin{minipage}[t]{0.48\columnwidth}
  \centering
  \includegraphics[width=\linewidth]{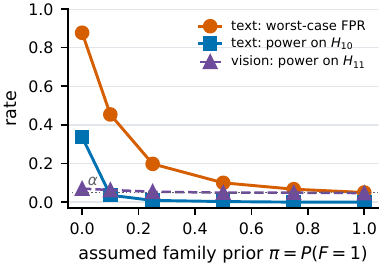}
  \caption{Sweeping the assumed family prior $\pi$ (\S\ref{sec:composite}). In text, worst-case FPR control
(orange) drives $H_{10}$ power (blue) to $0$ as $\pi\!\to\!1$; in vision, $H_{11}$ power (purple) stays
${\approx}\alpha$.}
  \label{fig:prior}
\end{minipage}
\end{figure}

\paragraph{Oracle four-world GLRT (not a monotone-threshold artifact).}
\label{app:glrt}
The auditors of \S\ref{sec:composite} threshold the scalar $\phi$. To rule out that a \emph{monotone} threshold is
simply the wrong test once there are four score distributions, we give the auditor a generalized likelihood-ratio
test with every advantage. We fit per-world multivariate Gaussians $p(z\mid H_{00}),\dots,p(z\mid H_{11})$ on a
feature vector $z$ (loss $\phi$ and, for text, Min-K\%), and score the composite exact-record ratio
$\mathrm{LR}(z){=}[p(z\mid H_{10}){+}p(z\mid H_{11})]/[p(z\mid H_{00}){+}p(z\mid H_{01})]$ (equal $F$ prior), fitting
the densities on half the shadow seeds and evaluating the ratio's AUC on the \emph{held-out} half (so the number is
honest, not fit${=}$eval). The GLRT does \emph{not} change the picture. In \emph{vision} it separates $E$ at AUC
$0.55$ (vs.\ $0.53$ for the monotone threshold), still chance, so the exact bit is genuinely \emph{absent} rather
than hidden by a poor threshold; in \emph{text} it reaches $0.74$ (vs.\ $0.69$; the Min-K\% feature adds the
$0.02$), confirming that text \emph{retains} exact information but that recovering it requires resolving the latent
family state. The substitutive/additive split is thus robust to the form of the test.

\paragraph{Partial family awareness is insufficient in text.}
\label{app:recall}
The family-conditioned auditor of \S\ref{sec:method} is an oracle that reproduces a target's \emph{entire}
near-duplicate family in its calibration shadows; a realistic auditor reproduces only part of it. We sweep the
number $m$ of siblings the auditor's shadows reproduce (pythia-2.8b/CC-News natdup at $\max$-sib${=}m$; $m{=}8$ is
the deployed full family), calibrate its family-present threshold there at $\alpha{=}5\%$, and measure the
false-attribution rate and power on the \emph{deployed} full-$8$-sibling model (exact-given-family AUC $0.90$).
False attribution stays far above $\alpha$ until the auditor reproduces at least half of the family.
The FPR is $0.70/0.19/0.06/0.05$ at $m{=}1/2/4/8$ over $12$ seeds per setting.
Power on true members is $0.99/0.84/0.56/0.50$, and $\mathrm{LR}^{+}$ is $1.40/4.50/8.80/10.00$.
Because the exact bit is
\emph{additive} in text (\S\ref{sec:verdict}), \emph{full} family recall controls the FPR while retaining
power, which is the sound family-conditioned outcome; in vision, power stays ${\approx}\alpha$ at every $m$.
Nominal calibration, however, requires reproducing at least half of the family ($m{=}4$), and enumerating every
natural duplicate of a target is infeasible at corpus scale. A realistic text audit therefore operates at low $m$
and false-attributes. This is the quantitative form of the abstention argument (\S\ref{sec:abstain}). Min-K++
calibration is weaker and noisier (deployed AUC $0.69$), so we report the suffix-loss auditor.

\section{Experiment map and reconciliation of reported leakage numbers}
\label{app:expmap}
The exact-record signal is reported under several constructions, and the numbers are \emph{not} interchangeable.
Table~\ref{tab:expmap} maps every headline value to the three canonical comparisons of \S\ref{sec:method}
(\emph{standard exact}, \emph{exact-given-family}, \emph{family presence}), plus \emph{content presence}
$C{=}E{\vee}F$ ($H_{00}$ vs.\ the other three worlds), and to its
measurement coordinates, namely target subset (vulnerable-$24$ = top-$50\%$ by discovery advantage; all-$48$ = every
shielded target), faithful-coverage $k$, the score, and the estimator (a \emph{pooled} one-threshold LiRA log-LR
audited by AUC, vs.\ a \emph{per-target} advantage averaged over targets). This factorization also positions the
audit relative to recent proposals that respond to near-duplicates by relaxing membership itself to a semantic
neighborhood~\citep{MHLMFJPC25}: rather than replacing the exact-record question with a coarser one, the
four-world design reports the exact bit and the family bit separately, and \citet{MHLMFJPC25} show that even the
relaxed definition remains unreliable under dataset manipulation. Values \emph{within} a quantity differ
by these coordinates and are not meant to coincide numerically; they should only share sign and order of magnitude. In
particular the four vision \emph{exact-given-family} numbers ($0.11$, $0.03$, $0.14$, $0.09\!\to\!0.01$) are one
near-chance quantity seen through four lenses (subset, pooling, score, and $k$), not a discrepancy; the single text
\emph{exact-given-family} value ($0.50$) sits an order of magnitude above all of them, reflecting the substitutive/additive split
(\S\ref{sec:verdict}). ``Advantage'' is $2|\mathrm{AUC}-0.5|$; where a construction is naturally read as a raw AUC
(the four-world audits) we give AUC and its advantage in parentheses.

\begin{table}[t]
\centering\footnotesize
\caption{Experiment map. Every reported exact/family leakage value, keyed to a canonical quantity ($E$=exact-record
bit, $F$=family presence) and its measurement coordinates. Arrows are defense effects (unshielded$\to$shielded, or
coverage $k$ swept). The numbers are consistent once the coordinates are read. For example, vision exact-given-family is
near chance under every lens, while text retains an order-of-magnitude-larger signal. Vision uses
CIFAR-10/ResNet-9, and text uses pythia-2.8b/CC-News with the suffix NLL score.}
\label{tab:expmap}
\resizebox{\textwidth}{!}{%
\begin{tabular}{@{}lllll@{}}
\toprule
Canonical quantity & value & modality, subset ($k$) & score, estimator & appears \\
\midrule
\textbf{Standard exact} & $0.55$ & vision, vuln-24 & LiRA log-LR, per-target adv. & \S\ref{sec:vision} (baseline) \\
$H_{10}$ vs.\ $H_{00}$ & $0.35$ & vision, all-48 & confidence, per-target adv. & Fig.~\ref{fig:f12} (unshielded) \\
\;($x$ alone, family absent) & $0.99$ & text, all shielded & suffix NLL, per-target adv. & \S\ref{sec:scaleup} \\
& $0.74$ (AUC $0.87$) & text & suffix NLL, pooled AUC & Tab.~\ref{tab:fourworld} \\
\midrule
\textbf{Exact-given-family} & $0.11$ & vision, vuln-24 ($k{=}8$) & LiRA log-LR, per-target adv. & \S\ref{sec:vision} (neardup) \\
$H_{11}$ vs.\ $H_{01}$ & $0.03$ (AUC $0.52$) & vision, pooled ($k{=}8$) & LiRA log-LR, pooled AUC & Tab.~\ref{tab:fourworld}, \S\ref{sec:verdict} \\
\;(family present, & $0.14$ ($0.15$ in App.~\ref{sec:whykeep}) & vision, all-48 & confidence, per-target adv. & Fig.~\ref{fig:f12} (shielded) \\
\;\;toggle $x$) & $0.09\!\to\!0.01$ & vision, all-48 ($k{=}1\!\to\!16$) & LiRA log-LR, pooled & Fig.~\ref{fig:frontier} (frontier) \\
& $0.50$ (AUC $0.75$) & text & suffix NLL, pooled AUC & Tab.~\ref{tab:fourworld}, \S\ref{sec:verdict} \\
\midrule
\textbf{Family presence} & $0.60$ & vision, all-48 & confidence, per-target adv. & Fig.~\ref{fig:f12} \\
$H_{01}$ vs.\ $H_{00}$ & $0.38$ (AUC $0.69$) & vision, all-48 & log-LR, pooled AUC & Tab.~\ref{tab:fourworld} \\
\;($x$ absent, toggle family) & AUC $0.98$ & text & suffix NLL, pooled AUC & Tab.~\ref{tab:fourworld}, \S\ref{sec:contentpresence} \\
 & $\mathrm{LR}^{+}\,3.20$ / $8.10$ & vision / text (WikiText) & log-LR verdict, $5\%$ FPR & App.~\ref{app:complete-results} \\
\midrule
\textbf{Content presence} & AUC $0.65$ & vision & log-LR, pooled AUC & this table only \\
$C{=}E{\vee}F$ & AUC $0.95$ & text & suffix NLL, pooled AUC & this table only \\
\bottomrule
\end{tabular}
}
\end{table}

\section{Text erasure frontier across fidelity and exposure}
\label{app:textfrontier}

Vision faithful copies \emph{erase} the exact-record verdict near-totally (aware $\mathrm{LR}^{+}\,1.00$, $88\%$ of
targets driven to advantage ${\le}0.05$ at $k{=}8$; \S\ref{sec:verdict}). Does the same hold in text, and is the
natural passages' non-erasure (\S\ref{sec:natdup}) merely their mid-fidelity? We answer both by sweeping sibling
\emph{fidelity} and \emph{exposure} independently on the $56$ WikiText targets, constructing optimization-free
near-verbatim siblings (a controlled fraction of tokens resampled from the corpus unigram distribution) at token
$5$-gram Jaccard $\in\{0.59,0.82,0.96\}$ and total exposure $E=n_{\text{sib}}\!\cdot\!K_{\text{shield}}\in\{2,8,16\}$,
always present while the exact target's membership is toggled; the auditor is the per-target aware LiRA of
\S\ref{sec:method} ($96$ shadows/arm, clean reference $\mathrm{LR}^{+}\,66.00$, advantage $0.90$).

\begin{table}[t]
\centering\small
\caption{Aware exact-record verdict on text as sibling fidelity (Jaccard) and exposure ($E$) vary.
WikiText/GPT-2 uses $\alpha{=}1\%$, with clean-reference $\mathrm{LR}^{+}\,66.00$ and advantage $0.90$.
The CIFAR $k{=}8$ reference has $\mathrm{LR}^{+}\,1.00$ and $88\%$ erased. Exposure, not fidelity, is the lever. At $E{=}2$, raising Jaccard
$0.59\!\to\!0.96$ barely moves the verdict; erosion appears only at high exposure and remains \emph{partial}
even at $16$ near-verbatim copies, and never reaches vision's near-total erasure.}
\label{tab:textfrontier}
\setlength{\tabcolsep}{5pt}
\begin{tabular}{@{}lccc@{}}
\toprule
 & $E{=}2$ & $E{=}8$ & $E{=}16$ \\
\midrule
Jaccard $0.59$ & $\mathrm{LR}^{+}\,58.00$ (adv $0.85$) & -- & -- \\
Jaccard $0.82$ & $\mathrm{LR}^{+}\,56.00$ (adv $0.81$) & $\mathrm{LR}^{+}\,23.00$ (adv $0.62$) & $\mathrm{LR}^{+}\,10.00$ (adv $0.39$) \\
Jaccard $0.96$ & $\mathrm{LR}^{+}\,65.00$ (adv $0.79$) & -- & $\mathrm{LR}^{+}\,7.00$ (adv $0.34$) \\
\bottomrule
\end{tabular}
\end{table}

Table~\ref{tab:textfrontier} shows the outcome. At matched low exposure ($E{=}2$) fidelity is nearly inert. Aware
$\mathrm{LR}^{+}$ moves only $58.00\!\to\!65.00$ across Jaccard $0.59\!\to\!0.96$, so the natural passages' survival is
\emph{not} explained by their fidelity; a near-verbatim sibling at the same exposure is no more hiding than a
mid-fidelity one. Erosion is driven by \emph{exposure}. Holding Jaccard $0.82$ fixed, $E{=}2\!\to\!16$ drives aware
$\mathrm{LR}^{+}\,56.00\!\to\!10.00$ (advantage $0.81\!\to\!0.39$). But even the extreme cell, $16$ copies of a
Jaccard-$0.96$ sibling, only \emph{partially} erodes the verdict ($\mathrm{LR}^{+}\,7.00$, advantage $0.16$--$0.34$,
$11\%$ of targets erased), far from vision's near-total collapse at $k{=}8$. Across every cell the \emph{unaware}
auditor is driven to fabrication (FPR ${\ge}0.85$). We read this as a substitutive/additive split. Autoregressive
suffix-loss memorization pins the \emph{exact} token sequence, so a distribution-aware auditor keeps a residual
grip on the true member that near-duplicates dislodge only slowly, whereas vision's boundary memorization is fully
supplied by faithful copies. In text the operative corruption of the exact-record verdict is therefore
\emph{fabrication} (universal, one to a few duplicates suffice), not \emph{erasure} (weak, exposure-gated, partial);
the fidelity-ordered mechanism (\S\ref{sec:model}) is shared across modalities but its reach into the exact-record
verdict is not.

\paragraph{A matched-exposure control shows that memorization is exposure-additive.} A naive fidelity sweep is confounded by
exposure, because inserting the target $K$ extra times makes it detectable by count alone. We therefore control for it, with one
target insertion against \emph{eight} sibling copies always present, varying only sibling fidelity. Even with eight
\emph{exact} copies present (the exact-copy multiplicity endpoint, $8$ vs.\ $9$), adding the exact instance once still yields exact-given-family advantage $0.74$
(WikiText/GPT-2, $85$ passages), barely below a different-suffix sibling ($0.83$) and matching the $0.75$ at
pythia-2.8b scale (\S\ref{sec:verdict}), far above vision's $0.03$ (\S\ref{sec:verdict}). Eight identical copies
do \emph{not} saturate the exact instance's contribution, and the per-token residual
$\Delta_t=\mathbb{E}[\ell_t\mid H_{01}]-\mathbb{E}[\ell_t\mid H_{11}]$ is spread roughly uniformly across the suffix
(mean $0.08$, flat; Fig.~\ref{fig:tokenresidual}) rather than concentrated at distinctive tokens. Autoregressive
suffix-loss memorization is thus \emph{exposure-additive}; near-duplicate coverage cannot substitute the exact
sequence, which is the mechanism behind the text-vs-vision asymmetry (vision's boundary memorization \emph{is}
saturated by faithful copies, collapsing the exact bit to chance).

\begin{figure}[t]
\centering
\includegraphics[width=0.5\columnwidth]{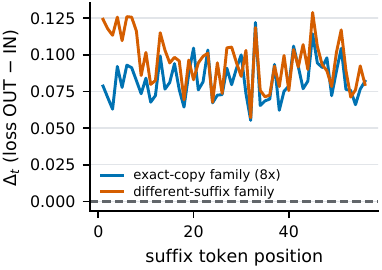}
\caption{Per-token residual $\Delta_t=\mathbb{E}[\ell_t\mid\text{target OUT}]-\mathbb{E}[\ell_t\mid\text{target
IN}]$ across the $56$ suffix positions, with eight sibling copies always present (matched exposure). The exact
target's presence lowers loss \emph{across the whole suffix} rather than at localized tokens. This distributed
occurrence-count signal is consistent with exposure-additive memorization.}
\label{fig:tokenresidual}
\end{figure}

\section{Details of the CC-News scale-up and deduplication}
\label{app:scaleup}

\paragraph{Setup.} We stream CC-News (the \emph{vblagoje/cc\_news} release), tokenize with the model tokenizer to fixed
$256$-token passages (exact duplicates removed), and mine natural near-duplicate families by token-$5$-gram
MinHash ($P{=}128$, $32$ bands) + LSH + union-find (\S\ref{sec:natdup} pipeline). The family-size distribution
over $40$k passages contains $1726$ families at Jaccard $0.5$, $921$ at $0.8$, $508$ at $0.9$ (median shielded
best-sibling Jaccard $0.89$). A \emph{standard} near-dedup~\citep{LINZECC22} removes exact duplicates
and near-duplicates at Jaccard ${\ge}0.8$; we keep the $58$ \emph{dedup-surviving} targets (best sibling in
$[0.5,0.8)$), $141$ \emph{dedup-removed} ($\ge0.8$), and $100$ no-sibling bystanders, plus a $6$k-passage
deduped background; one of the $200$ families has best-sibling Jaccard below $0.50$ and is excluded from the
banded analyses. We full-finetune pythia-2.8b (bf16 AMP, gradient checkpointing, AdamW $10^{-5}$, single
target exposure) over $128$ shadows ($64$ \emph{none} / $64$ \emph{natdup}), toggling each exact target's
membership while its natural siblings are always present in \emph{natdup}.

\paragraph{Calibrated verdict.} Per-target LiRA log-LR, threshold calibrated to $\alpha$-FPR on calibration OUT,
FPR measured out of sample with a crossed bootstrap over calibration seeds, evaluation seeds, and targets.

\begin{table}[h]
\centering\small
\setlength{\tabcolsep}{4pt}
\caption{Calibrated CC-News exact-record verdicts. Thresholds are calibrated on family-absent non-members; the unaware and aware auditors are evaluated with natural near-duplicate families present. LR$^{+}$ is the crossed-bootstrap median of per-resample TPR/FPR and need not equal the ratio of the printed point rates.}
\label{tab:ccnews-verdict}
\begin{tabular}{@{}llccc@{}}
\toprule
$\alpha$ & auditor & TPR & FPR & $\mathrm{LR}^{+}$ \\
\midrule
$1\%$ & (A) clean $\to$ clean & $1.00$ & $0.02$ & $\mathbf{61.67}$ \\
$1\%$ & (B) unaware $\to$ natdup & $1.00$ & $\mathbf{1.00}$ & $\mathbf{1.00}$ \\
$1\%$ & (C) aware $\to$ natdup & $0.64$ & $0.02$ & $28.36$ \\
$5\%$ & (B) unaware $\to$ natdup & $1.00$ & $1.00$ & $1.00$ \\
$5\%$ & (C) aware $\to$ natdup & $0.78$ & $0.06$ & $13.69$ \\
\bottomrule
\end{tabular}
\end{table}

The unaware verdict is fabricated ($\mathrm{LR}^{+}\,{\to}\,1$); the aware auditor retains power ($\mathrm{LR}^{+}\,28$,
paired advantage $0.90\!\to\!0.75$, only $2\%$ of targets driven below advantage $0.05$). Text is therefore not erased
(App.~\ref{app:textfrontier}), but awareness requires infeasible enumeration. Family presence stays a usable
verdict ($\mathrm{LR}^{+}\,9.1$ at $\alpha{=}1\%$, $7.3$ at $5\%$; bootstrap-median). A four-world audit (over $\{$x-only, dup-only, both, neither$\}$)
recovers family presence (dup-only vs.\ neither, AUC $0.98$) far better than exact inclusion given the family
(both vs.\ dup-only, AUC $0.75$). The exact bit is thus partially recoverable, unlike vision's near-chance $0.52$
(\S\ref{sec:verdict}), which is the substitutive/additive split again.

\paragraph{Fabrication by dedup band.} A realistic unaware auditor's false-``member'' rate on genuine
non-members, by whether a standard MinHash dedup ($\text{Jaccard}\,{\ge}\,0.80$) would remove the siblings. The
\emph{surviving} band is not removed by dedup yet still fabricates; bystanders (no siblings) are the control.

\begin{table}[h]
\centering\small
\setlength{\tabcolsep}{5pt}
\caption{False-member rates for CC-News exact non-members by natural-family Jaccard band. The threshold is calibrated on family-absent non-members at each nominal FPR $\alpha$.}
\label{tab:ccnews-jaccard-bands}
\begin{tabular}{@{}lcccc@{}}
\toprule
target band & base adv. & $\alpha{=}1\%$ & $\alpha{=}5\%$ & $\alpha{=}10\%$ \\
\midrule
\textbf{dedup-surviving} ($J\,0.50$--$0.80$) & $0.98$ & $\mathbf{0.55}$ & $\mathbf{0.85}$ & $\mathbf{0.91}$ \\
dedup-removed ($J{\ge}0.80$)               & $0.99$ & $0.63$ & $0.89$ & $0.95$ \\
bystander (no siblings, control)          & $1.00$ & $0.00$ & $0.01$ & $0.03$ \\
\bottomrule
\end{tabular}
\end{table}

Min-K\%~\citep{SAXHLBCZ24} reproduces the surviving-band fabrication ($0.59$/$0.82$ at $\alpha{=}1\%$/$5\%$),
so it is not an artifact of the loss score.

\paragraph{Neighbourhood-comparison attack (\S\ref{sec:scaleup}).} The most semantically relevant baseline for text MI~\citep{MMJSSB23} scores $\phi_{\mathrm{nbr}}(x){=}\overline{\mathrm{NLL}(\text{$8$ Qwen paraphrases of }x)}-\mathrm{NLL}(x)$.
This score normalizes the suffix loss with the target's \emph{own} paraphrase set.
On matched four-world shadows with $6$ \emph{none} and $6$ \emph{natdup} seeds and $1{,}800$ scores each, it does not escape the confound.
Under a common pooled clean-reference calibration at $\alpha{=}5\%$, its false attribution on family-present non-members is $0.58$ versus $0.99$ for the raw loss.
The result is lower but still ${\gg}\alpha$.
The lexical family lowers $\mathrm{NLL}(x)$ but not the loss of the untrained paraphrase neighbors, so normalization removes only part of the shift.
It also weakens the exact-given-family signal from AUC $0.90$ to $0.58$, while standard family-absent MI stays strong at $1.00$ to $0.92$.
Normalizing by neighbors trades exact-record power for false attribution that it only partly suppresses.

\paragraph{Executed dedup-and-retrain.} The partition above holds the training set fixed. As the \emph{direct} test
we instead \emph{remove} every best-sibling-Jaccard-${\ge}0.8$ family from the corpus and \emph{re-finetune}
pythia-2.8b on the resulting $\tau{=}0.8$-deduplicated corpus ($64$ \emph{natdup} shadows; the family-absent worlds
are unchanged, so the \emph{none} shadows are reused). Fabrication does not abate; it \emph{sharpens}, because the
confounding high-Jaccard families of \emph{other} targets are now absent from training. The \emph{surviving}
families ($J{<}0.8$, which a $\tau{=}0.8$ pass keeps) are false-attributed at $0.92$/$0.99$ (5\% FPR,
$[0.50,0.70)$/$[0.70,0.80)$ bands; $0.96$/$1.00$ at $10\%$), \emph{above} the fixed-corpus partition estimate
($0.85$), while the exact-given-family signal is \emph{retained and if anything sharpened} (four-world AUC $0.90$
on the full corpus $\to$ $0.97$ after the $\tau{=}0.8$ dedup, over the $200$ CC-News families of this
strong-memorization finetune; naive $H_{10}$ vs $H_{00}$ $\approx1.0$). The main-text $0.75$ (Tab.~\ref{tab:fourworld})
is a \emph{separate}, weaker-memorization pythia run (naive $0.87$); exact-given-family tracks how strongly the
finetune memorizes, but is well above chance in every run.
Even the \emph{removed} band ($J{\ge}0.8$, whose closest
siblings are deleted) still fabricates at $0.66$--$0.71$ (5\% FPR) from residual sub-threshold near-duplicates. So
the failure survives an \emph{executed} deduplicate-and-retrain, not merely a measured-Jaccard partition of a fixed
model.

\paragraph{LoRA ablation.} To confirm the effect is not specific to full finetuning of pythia, we repeat the test
with LoRA finetuning ($r{=}16$, $10.9$M trainable parameters) of a different-architecture $2$B model (Qwen3.5-2B), on
the same CC-News families ($48$ shadows). Baseline exact-record MI is again strong (advantage $1.00$), and the
surviving-band fabrication reproduces. The unaware auditor's false-``member'' rate on dedup-surviving non-members
rises to $0.16$/$0.37$/$0.77$ at $\alpha{=}1\%$/$5\%$/$10\%$ (removed band $0.16$/$0.39$/$0.78$), while bystanders
stay clean ($0.00$). The magnitude is lower than pythia's full finetuning at the $1\%$ operating point, but the
effect, dedup-surviving-band fabrication with a clean bystander control, holds across architecture and
finetuning method.

\section{Semantic-family experiments, the objective control, and the lexical confound}
\label{app:semantic}
\paragraph{Paraphrase families.} To decouple \emph{meaning} from \emph{tokens} we build \emph{semantic} families.
For each shielded target we detokenize it, prompt Qwen2.5-7B (few-shot) for candidate rewrites that preserve the
facts, and keep the $k{=}8$ with the \emph{lowest} word-$5$-gram Jaccard to the target. Mean lexical Jaccard is
$0.00$ (AG News) / $0.01$ (CC-News), versus ${\approx}0.9$ for the near-duplicate families. They share the topic but almost
no $n$-grams.

\paragraph{Objective control (\S\ref{sec:objective}).} On AG News we hold the backbone (GPT-2), records, semantic
family, and exposure fixed and swap \emph{only} the objective, comparing GPT-2 with a sequence-classification head against the
same GPT-2 as an autoregressive LM, each in the four-world design ($64$ shadows). Both have a strong family-absent exact
signal (naive $H_{10}$ vs $H_{00}$, classification AUC $0.61$, LM $1.00$). They diverge on \emph{exact-given-family}
($H_{11}$ vs $H_{01}$), with classification at $0.53$ (chance; $R_{\mathrm{exact}}{=}0.27$, $95\%$ CI $[0.19,0.36]$) and the LM at
$1.00$ ($R_{\mathrm{exact}}{=}0.86$, $[0.82,0.89]$). We read the split from AUC and the standardized
$R_{\mathrm{exact}}$ rather than the raw interaction $I$ (both ${\approx}{-}2.5$), because the LM's suffix NLL saturates at a
floor for memorized targets, so the raw mean interaction is a ceiling artifact while identifiability is not. With
everything but the objective fixed, classification erases the exact bit and the LM retains it, so the objective is
causal.

\paragraph{Lexical vs.\ semantic families (\S\ref{sec:objective}).} Re-running the pythia-2.8b four-world with the
CC-News \emph{semantic} paraphrase family in place of the lexical one (natdup arm re-finetuned; family-absent worlds
reuse the \emph{none} shadows; $28$ shadows) collapses the clean-reference fabrication to $\mathrm{FPR}_{H_{01}}{=}0.10$
(${\approx}\alpha$), versus $0.88$ for the lexical family, while exact-given-family stays $1.00$. Paraphrases that
share meaning but not tokens barely move the target's suffix loss in this single-exposure regime, so they neither
fabricate nor mask; the residual $0.10$ is still twice the nominal $\alpha$. In this regime the confound is
predominantly \emph{lexical} and the recovered quantity is near-duplicate \emph{family} presence. The AG-News
objective control (above) shows the boundary is regime-dependent, because under multi-epoch fine-tuning the same
paraphrase construction is detectable at family-presence AUC $0.89$ (\cref{tab:fourworld}).
Concurrent work reaches the complementary conclusion from the adversarial side: \citet{EBCMD26} show that an
accused developer can deliberately paraphrase training data to degrade state-of-the-art MIAs while preserving
utility, with residual leakage that varies across fine-tuning regimes. Their deliberate obfuscation and our
natural and controlled paraphrase families probe the same lexical dependence of the membership signal, and both
observe the regime dependence documented in \S\ref{sec:objective}.

\paragraph{Natural family features and the confound.} The controlled fidelity/$k$ sweeps
(\S\ref{sec:mechanism}) have a natural analogue. On the CC-News families, the OUT-score shift rises monotonically
with best-sibling fidelity (mean shift $1.19/1.22/1.35/1.38$ across the
$[0.50,0.70)$/$[0.70,0.80)$/$[0.8,0.9)$/$[0.9,1.0)$ Jaccard bands; Spearman $+0.19$ of best-sibling Jaccard with
the shift), while false attribution is high in every band (false-``member'' rate $0.80/0.88/0.82/0.92$ at
$\alpha{=}5\%$). Family size adds to the shift once ${\ge}2$ siblings are present.

\section{Additional quantitative results from the full study}
\label{app:complete-results}

This section records the secondary operating points, confidence intervals, and construction statistics from the
full result presentation. They are included for completeness and do not change the headline comparisons in the
main text.

\paragraph{Study scale and natural-family construction.}
For the primary controlled CIFAR-10 experiment, the $48$ targets and $k{=}8$ family records add $384$ records,
or approximately $0.8\%$ of the base training set. Across all experiments, the study trains approximately
$1{,}500$ shadow models; the per-experiment allocation is reported in Table~\ref{tab:shadow-counts}.
For WikiText-103, target-to-sibling token-$5$-gram Jaccard similarity has median $0.60$ and range
$[0.40,0.81]$ after manual verification that the passages contain copied text.

\paragraph{Complete controlled-vision verdict.}
All five decision-level rows of the controlled CIFAR-10 verdict appear in \cref{tab:vision-verdict}. In
particular, the clean-reference auditor evaluated on a conflicting-label family has TPR $0.13$, FPR $0.02$, and
$\mathrm{LR}^{+}{=}5.4$.
On CIFAR-10, a fixed-threshold family verdict reaches $\mathrm{LR}^{+}{=}3.20$ with TPR $0.57$ at nominal $5\%$ FPR.

Under the family-conditioned faithful comparison, per-target advantage declines by $0.37$, with crossed-bootstrap
$95\%$ CI $[0.20,0.52]$, and $88\%$ of the $24$ vulnerable targets fall below advantage $0.05$.
This family-conditioned $0.37$ is not interchangeable with the clean-reference $-0.44$ of \S\ref{sec:trich} or
the independent budget sweep's $-0.43$ at $k{=}8$; the three constructions differ in comparison world and
training run and are expected to agree only in sign and magnitude (App.~\ref{app:expmap}).
In the four-world mean decomposition, family presence changes the vision exact-record effect from
$\delta_E^0{=}1.08$ to $\delta_E^1{=}0.18$; the interaction is $I{=}{-}0.90$ with $95\%$ CI
$[-1.46,-0.32]$, and $R_{\mathrm{exact}}{=}0.23$. WikiText is closer to additive, with
$I{=}{-}0.02$, CI $[-0.06,0.03]$, and $R_{\mathrm{exact}}{=}0.89$; the pythia-2.8B experiment gives
$I{=}{-}0.36$ and $R_{\mathrm{exact}}{=}0.60$.
The conflicting-label arm increases advantage from $0.55$ to $0.79$, a change of $+0.24$.

\paragraph{Mechanism, budget, and generalization details.}
The target's two nearest same-class training neighbors receive an OUT-score shift of $+0.20$, with $95\%$ CI
$[0.05,0.34]$, while their membership advantage changes by only $+0.01$, CI $[-0.02,0.05]$.
In the near-duplicate budget sweep, the advantage change grows from $-0.23$ at $k{=}1$ to $-0.43$ at
$k{=}8$ and saturates at $-0.39$ for $k{=}16$; the kernel model fitted on $k\leq8$ predicts $-0.44$ at
$k{=}16$. With online augmentation, the reinforcement-variant effect weakens from $-0.24$ to $-0.08$.

Before vulnerable-target selection, CIFAR-100 has baseline advantage $0.73$. Near duplicates change the selected
targets' advantage by $-0.76$, the density control remains null, conflicting labels change it by $+0.07$, and
bystanders remain unchanged. On CIFAR-10/VGG-11, near duplicates and conflicting labels change advantage by
$-0.30$ and $+0.30$, respectively. In controlled language modeling, copies and near duplicates change advantage
by $-0.44$ and $-0.50$ for the model trained from scratch, and by $-0.32$ and $-0.11$ for finetuned GPT-2.
Conflicting evidence is null in these text settings, with changes of $-0.09$ and $0.00$.

\paragraph{Complete WikiText operating points.}
With natural siblings present, the family-conditioned auditor retains exact-record advantage approximately
$0.99$ and $\mathrm{LR}^{+}{=}97.00$. At nominal $1\%$ FPR, the clean-reference false-member rate rises from
$1.4\%$ to $81\%$, with $95\%$ CI $[71\%,90\%]$; no-sibling controls have FPR $1.9\%$.
At nominal $5\%$ FPR, the family-present false-member rate is $85\%$ and the positive-verdict
$\mathrm{LR}^{+}$ falls to $1.20$ (from $66.00$ at the clean reference's $1\%$ operating point). Across retraining environments, the observed false-attribution
magnitude ranges from approximately $28\%$ to $85\%$.
With the exact passage absent, toggling family presence gives advantage $1.00$, versus $0.07$ for no-sibling
controls. Its fixed-threshold family verdict has $\mathrm{LR}^{+}{=}8.10$, TPR $0.88$, and realized FPR $0.11$
at nominal $5\%$ FPR.

\paragraph{Complete CC-News operating points.}
The family-absent exact-record advantage is $0.99$. With natural families present, the clean-reference
false-member rate is $99.7\%$ and $\mathrm{LR}^{+}$ falls to $1.00$, compared with $19.60$ in the clean world at
$5\%$ FPR and approximately $62.00$ at $1\%$ FPR. At nominal $5\%$ FPR, the four original best-sibling Jaccard
bands $[0.50,0.70)$, $[0.70,0.80)$, $[0.80,0.90)$, and $[0.90,1.00)$ have false-member rates
$80\%$, $88\%$, $82\%$, and $92\%$, respectively; no-sibling controls remain at $1.2\%$.
Across these bands, family-presence AUC lies between $0.92$ and $0.97$, while the clean-reference
false-member rate lies between $0.59$ and $0.94$.
A fresh pythia-2.8B run gives $98\%$, Min-K\% gives $97\%$, Min-K++ gives $82\%$, and EZ-MIA gives $63\%$
false attribution at the stated $5\%$ operating point. Across the four Jaccard bands, the Min-K variants range
from $86\%$ to $100\%$.

\paragraph{Information-relocation frontier.}
Across $k{=}1$ to $16$, the exact-given-family point estimate falls from $0.09$ to $0.01$ while the
family-presence estimate rises from $0.12$ to $0.48$. The smallest observed maximum of the two advantages is
$0.12$ at $k{=}1$. Neighboring bootstrap intervals overlap, so this finite sweep supports a descriptive
signal-redistribution pattern; it does not establish a precise monotone law or a universal lower bound.
In $k{=}1,2,4,8,16$ order, crossed-bootstrap $95\%$ CIs are respectively
$[0.01,0.21]$, $[0.00,0.17]$, $[0.00,0.15]$, $[0.00,0.12]$, and $[0.00,0.14]$ for exact-given-family
advantage, and $[0.01,0.24]$, $[0.10,0.34]$, $[0.15,0.41]$, $[0.24,0.52]$, and $[0.33,0.62]$ for
family-presence advantage. The Pythia reference intervals are $[0.45,0.56]$ and $[0.93,0.97]$.

\end{document}